\documentclass[journal]{IEEEtran}

\usepackage[utf8]{inputenc} 
\usepackage[T1]{fontenc}
\usepackage{amssymb}
\usepackage{url}
\usepackage{ifthen}
\usepackage {tikz}
\usepackage{pdfpages}
\usepackage{caption}
\usepackage{cite}
\usepackage{amsthm}
\usepackage{algorithmic}
\usepackage{graphicx}
\usepackage{textcomp}
\usepackage{xfrac}
\usepackage{xcolor}
\usepackage{comment}
\usepackage{float}
\usepackage{bbm}
\usepackage{scalerel}
\usepackage{ragged2e}
\usepackage{enumitem}
\usepackage{multirow}
\usepackage{orcidlink}
\newcounter{hints}
\renewcommand{\thehints}{\roman{hints}}
\newcommand{\hintedrel}[2][]{%
  \refstepcounter{hints}%
  \if\relax\detokenize{#1}\relax\else\label{#1}\fi
  \mathrel{\overset{\mathrm{(\thehints)}}{\vphantom{\le}{#2}}}%
}
\newcommand{\restarthintedrel}{\setcounter{hints}{0}}
\usepackage{tikz}

\usetikzlibrary{dsp,chains}

\usepackage{comment}
\usepackage{blindtext}
\usepackage{balance}

\usetikzlibrary {positioning}
\usepackage{tkz-euclide}

\usepackage[cmex10]{amsmath} 

\newcommand\NB[1][0.3]{N\kern-#1em\textcolor{red}{B}}

\DeclareMathOperator*{\given}{\vert}
\newcommand{\E}[1]{\mathsf{E}\left[#1\right]}
\newcommand{\Var}[1]{\mathsf{Var}\left(#1\right)}
\newcommand{\Varvec}[1]{\boldsymbol{\mathsf{Var}}\left(#1\right)}

\newcommand{\snr}{\mathsf{snr}}
\newcommand{\mmse}[1]{\mathsf{mmse}\left(#1\right)}
\newcommand{\loge}[1]{\log\left(#1\right)}
\newcommand{\e}[1]{e^{#1}}
\newcommand{\abs}[1]{\left\vert#1\right\vert}

\newcommand\labelrel[2]{%
  \begingroup
    \refstepcounter{relctr}%
    \stackrel{\textnormal{(\alph{relctr})}}{\mathstrut{#1}}%
    \originallabel{#2}%
  \endgroup
}
\AtBeginDocument{\let\originallabel\label} 

\newtheoremstyle{note}
  {}
  {}
  {\itshape}
  {}
  {\itshape}
  {}
  {0em}
  {\thmname{#1}\thmnumber{ #2}}

\theoremstyle{note} 
\newtheorem{theorem}{Theorem}
\newtheorem{corollary}{Corollary}
\newtheorem{lemma}{Lemma}
\newtheorem{proposition}{Proposition}
\newtheorem{remark}{Remark}

\newcounter{relctr} 

\providecommand{\revision}[1]{{\color{black} #1}}

\begin{document}
\title{Differential Entropy of the Conditional Expectation under Additive Gaussian Noise} 

\author{{Arda~Atalik$^{\orcidlink{0000-0003-3439-7838}}$, ~Alper~Köse$^{\orcidlink{0000-0002-1537-1205}}$,~and~Michael~Gastpar$^{\orcidlink{0000-0002-5499-5336}}$, \IEEEmembership{Fellow, IEEE}}
\thanks{The work in this manuscript was supported in part by the Swiss National Science Foundation under Grant 200364. The work in this manuscript was partially presented at the {\it 2021 IEEE Information Theory Workshop,} Japan.}
\thanks{A. Atalik is with the Department of Electrical and Electronics Engineering, Bilkent University, 06800 Ankara, Turkey (email: arda.atalik@bilkent.edu.tr).}
\thanks{A. Köse is with the Electrical and Electronics Engineering Department, Bogazici University, 34470 Istanbul, Turkey (email: alper.kose@boun.edu.tr).}
\thanks{M. Gastpar is with the School of Computer and Communication Sciences, {\'E}cole Polytechnique F{\'e}d{\'e}rale de Lausanne (EPFL), Lausanne, Switzerland (email: michael.gastpar@epfl.ch).}
}
\maketitle
\begin{abstract}
The conditional mean is a fundamental and important quantity whose applications include the theories of estimation and rate-distortion. It is also notoriously difficult to work with. This paper establishes novel bounds on the differential entropy of the conditional mean in the case of finite-variance input signals and additive Gaussian noise. The main result is a new lower bound in terms of the differential entropies of the input signal and the noisy observation. 
The main results are also extended to the vector Gaussian channel and to the natural exponential family.
Various other properties such as upper bounds, asymptotics, Taylor series expansion, and connection to Fisher Information are obtained. Two applications of the lower bound in the remote-source coding and CEO problem are discussed. 
\end{abstract}

\section{Introduction and Motivation}


The conditional expectation is a quantity of fundamental interest with myriad applications.
It is an intuitively pleasing estimator of an underlying signal, given a noisy observation. For example, it is well known that subject to a mean-squared error (MSE) criterion,
the conditional expectation is the optimal estimator. More generally, the conditional expectation is a sufficient statistic for a large class of problems. \revision{Applications of conditional expectation thus include detection and estimation \cite{predictionOptimality, estimationBlackwell,bregman2}, multiterminal hypothesis testing \cite{han1987hypothesis,zaidi_2020}, information bottleneck \cite{tishby2000information,zaidi2020information}, privacy funnel \cite{makhdoumi2014information,calmon2015fundamental}, equalization in communication systems, and statistical physics.}

While the conditional mean has a straightforward formula, it is usually impossible to express it in closed form.
This hampers the exploration of its key properties, such as its moments or its differential entropy.
\revision{The goal of the present investigation is to shed light on the {\it differential entropy} of the conditional mean $h(\E{X\given Y}).$ This quantity has a multitude of applications. We present and develop one of these applications, which concerns problems of noisy (or {\it remote}) source coding, where the encoding device does not get to observe the source $X$ of interest, but rather only a noisy version $Y$ of the source.}

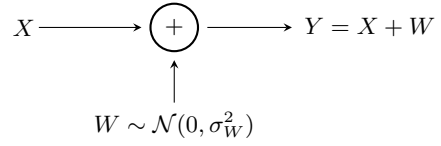
\begin{figure}[t]
\begin{center}
\small
\begin{tikzpicture}[scale=2,shorten >=1pt, auto, node distance=1cm,
   node_style/.style={scale=1,circle,draw=black,thick},
   edge_style/.style={draw=black,dashed}]

    \node [fill=none] at (-1,0) (nodeS) {$X$};
    \node [fill=none] at (0,-0.65) (nodeS) {$W \sim \mathcal{N}(0,\sigma_W^2)$};
    \node [fill=none] at (1.3,0) (nodeS) {$Y = X + W$};
    \node[node_style] (v1) at (0,0) {$+$};
   
    \draw [-stealth](-0.9,0) -- (-0.2,0);
    \draw [-stealth](0,-0.5) -- (0,-0.2);
    \draw [-stealth](0.22,0) -- (0.8,0);
    \end{tikzpicture}
\end{center}
\caption{The additive Gaussian noise observation model}
\label{observation}
\end{figure}

\par In this paper, our primary focus concerns the additive Gaussian noise model
\begin{equation}
    Y = X+W
    \label{themodel}
\end{equation}
where $W$ is a zero-mean Gaussian random variable of variance $\sigma_W^2,$ independent of the signal $X,$
and where $X$ has an arbitrary distribution. This is illustrated pictorially in Figure~\ref{observation}.
Estimation in Gaussian noise has been a central topic at the intersection of estimation and information theory for the last sixty years.
Consequently, in this special case, the conditional expectation $\E{X\given Y}$ has been studied in a wealth of works and many key tools are known. A first such tool that will be of importance to our derivations is \emph{Tweedie's Formula}~\cite{robbins2020empirical}, which connects the conditional mean to the score function. Following up on this important formula, the \emph{Hatsell-Nolte Identity} was discovered~\cite{hatsell_nolte_1971}, which relates the conditional mean and variance. Yet another key formula is \emph{Brown's Identity}~\cite{brown1971admissible}, which shows a connection between the minimum mean-square error (MMSE) and Fisher Information.

\revision{It is also interesting to compare and contrast $h(\E{X\given Y})$ to a more common quantity, namely, the conditional entropy of $X$ given $Y,$ usually denoted as $h(X|Y).$ 
While a thorough study is outside of the scope of the present paper, we may observe that the two quantities behave quite differently.
Starting from the Gaussian noise model of Equation~\eqref{themodel}, one may consider the limiting cases as $\sigma_W^2$ tends to either of its extremal values. First, as $\sigma_W^2$ tends to zero, we can observe that $h(\E{X\given Y})$ simply tends to $h(X)$ while $h(X|Y)$ diverges to $-\infty.$ On the other end of the scale, as $\sigma_W^2$ tends to $\infty,$ we observe that $h(\E{X\given Y})$ diverges to $-\infty$ while $h(X|Y)$ tends to $h(X).$ In the case where $X$ is Gaussian, too, one can derive closed-form expressions for both quantities, as follows:
\begin{align}
    h(X\given Y) &= h(X)+h(Y\given X)-h(Y)\nonumber\\
    &= h(X)+h(W)-h(Y)\nonumber\\
    &= \frac{1}{2}\loge{2\pi e \frac{\sigma_X^2\,\sigma_W^2}{\sigma_X^2+\sigma_W^2}}\\
    h(\E{X\given Y}) &= h\left(\frac{\sigma_X^2}{\sigma_X^2+\sigma_W^2}Y+\frac{\sigma_W^2}{\sigma_X^2+\sigma_W^2}\mu_X\right)\nonumber\\
    &= h(Y)+\log \frac{\sigma_X^2}{\sigma_X^2+\sigma_W^2}\nonumber\\
    &= \frac{1}{2}\loge{2\pi e \frac{\sigma_X^4}{\sigma_X^2+\sigma_W^2}}.
\end{align}
Finally, we would like to emphasize the following \emph{dual} entropy power inequalities, which are valid for arbitrary inputs \cite[Sec.~8.6]{cover_thomas_2006}:
\begin{subequations}
\begin{align}
    N(X\given Y)&\leq \E{\Var{X\given Y}} \label{nxgiveny}\\
    N(\E{X\given Y})&\leq \Var{\E{X\given Y}}\label{nexgiveny},
\end{align}
\end{subequations}
and satisfied with equalities for Gaussian $X$.
}

\subsection{Contributions}
\begin{itemize}
    \item We provide a new lower bound which relates the differential entropy of the conditional mean to that of input and output in Figure~\ref{observation}: 
\begin{equation}
    h(\E{X\given Y}) \ge 2h(X)-h(Y).
    \label{mres}
\end{equation}
\item We also derive novel upper bounds on $h(\E{X\given Y})$ and show several applications of the new bounds, most notably in the context of so-called remote source coding.
\item We extend our bound to the case where the noisy observation process is not characterized by adding Gaussian noise, but by a general exponential family distribution.
\item We also extend our bounds to the case of vector signals.
\end{itemize}

\subsection{Related Work}

\emph{Differential entropy} plays a key role in estimation theory as it measures the uncertainty in a given probabilistic scenario and it has been gaining importance in many different fields such as data science and machine learning, biology and neuroscience, economics, and other experimental sciences as highlighted in \cite{verdu2019empirical}.

\emph{Mean-squared error} (MSE) is one of the most commonly used error metrics in estimation theory and plays a significant role in many real-life applications ranging from signal detection in communications to regression problems in machine learning. Typically, in an estimation problem, the objective aims at minimizing a cost function, thus it is highly significant to find the estimator which leads to the \emph{minimum mean-square error} (MMSE). As widely known, the conditional mean is optimal in \revision{the} MSE sense, i.e., the MMSE estimate of $X$ observing $Y=y$ is found by the conditional mean of $X$ given $Y=y$. As a result, conditional mean is widely used in signal processing applications, e.g., signal detection \cite{tanahashi}, noise cancellation \cite{spagnolini}, frequency estimation \cite{james}, and target tracking \cite{miller}.  

In \cite{mmse_guo_1}, a fundamental derivative identity connecting the mutual information and MMSE is discovered and in \cite{palomar_verdu_2006}, it has been further explored.
Note that by contrast to the model in Fig. \ref{observation},  the model in \cite{mmse_guo_1} fixes the noise variance to unity and instead (but equivalently) lets the noisy observation be $Y=\sqrt{\gamma} X + W.$
Our results are more naturally expressed in terms of the model in Fig. \ref{observation}.
Properties of MMSE such as monotonicity, convexity, and infinite differentiability as a function of $\snr$ have been shown in \cite{mmse_guo_2}, while its functional properties as a function of input-output distribution have been analyzed in \cite{wu_verdu_2012,dytso_poor_bustin_shamai_2018}. Recently, in \cite{dytso_poor_shitz_2020,dytso_cardone}, the authors have focused on the derivatives of the conditional mean with respect to the observation, and many previously known identities in the literature have been recovered.



\subsection{Outline}
\begin{itemize}
    \item In Section~\ref{mainresultsection}, we provide a new lower bound which relates the differential entropy of the conditional mean to that of input and output in Figure~\ref{observation}.
\item In Section \ref{furthersection}, we study further properties of the differential entropy of \revision{the} conditional mean such as upper bounds, Taylor series expansion, low-and-high input variance asymptotics, and connection to Fisher Information. For different input distributions, the bounds on $h(
\E{X\given Y})$ are illustrated. 
\item In Section \ref{applicationsection}, two applications of the lower bound in the remote source coding \cite{eswaran_gastpar_2019} are investigated.
\item In Section \ref{extensiontoexponentialsection}, the main result is extended to the natural exponential families. 
\item In Section \ref{extensiontovectorsection}, the lower and upper bounds are extended to the vector additive Gaussian noise model.
\end{itemize}

\subsection{Notation}\label{notationsubsection}
We use uppercase letters $X,\,Y$ to denote random variables and lowercase letters $x,\,y$ to denote their realizations. In slight abuse of notation, we use  boldface uppercase letters $\mathbf{X},\,\mathbf{Y}$ to denote both random vectors as well as (deterministic) matrices. The distinction will be clear from context. Given a square-integrable, absolutely continuous random variable $X$ with a probability density function (PDF) $p_X(x)$, \revision{its mean $\E{X}$ is denoted as $\mu_X$}, its variance $\Var{X}$ is denoted as $\sigma_X^2$, and its differential entropy is 
\begin{equation}
    h(X) = -\int p_X(x)\,\log p_X(x)\,dx
\end{equation} 
where $\log$ denotes the natural logarithm.
The entropy power of $X$ is 
\begin{equation}
    N(X) = \e{2h(X)}/(2\pi e),
\end{equation} \revision{the conditional entropy power of $X$ given $Y$ is $N(X\given Y) = \e{2h(X\given Y)}/(2\pi e)$}, and the Fisher information of $X$ is $J(X) = \int p_X(x)\,\left(\frac{d}{dx}\,\log p_X(x)\right)^2 dx$. \revision{To denote $X$ is a Gaussian random variable with mean $\mu_X$ and variance $\sigma_X^2$, we use the notation $X\sim \mathcal{N}(\mu_X,\sigma_X^2)$.}

For two random variables $(X,Y)$ with a density function $p_{X,Y}(x,y),$ the
mutual information is  $I(X;Y) = h(Y)-h(Y\given X)=h(X)-h(X\given Y).$
Denote the conditional expectation and variance of $X$ given $Y$ as $\E{X\given Y}\text{ and }\Var{X\given Y}$, respectively, and the corresponding mean-square error as 
\begin{equation}
    \mmse{X\given Y} = \E{\Var{X\given Y}} = \E{(X-\E{X\given Y})^2}.
\end{equation}

It is conceptually straightforward to give an expression for the probability density function (PDF) of the conditional mean $\E{X\given Y},$ and thus, for the differential entropy of $\E{X\given Y}.$ Some simplifications can be applied in the special case of additive Gaussian noise considered in the present paper, see e.g.~\cite{dytso_poor_shitz_2020}, but in general, the resulting expressions are intractable. Therefore, in our work, we will not leverage the probability density function of $\E{X\given Y}$ directly. 

\revision{\subsection{The Case of Gaussian Inputs}\label{v2addedsubsection}
In this section, we briefly review the well-known formulas for the case where the underlying source $X$ is Gaussian.
For this case, all the quantities defined in Subsection~\ref{notationsubsection} can be calculated analytically. The probability density function of $X\given Y=y$ can be calculated by Bayes' rule to find
\begin{align}
    f_{X\given Y}(x|y) &= \frac{p_{Y|X}(y|x)\, p_{X}(x)}{p_{Y}(y)} \nonumber\\
    &= \frac{p_{W}(y-x)\, p_{X}(x)}{p_{Y}(y)}\nonumber\\
    &= \frac{\exp\left({-\frac{\left(x-\left(\frac{\sigma_X^2}{\sigma_X^2+\sigma_W^2}y+\frac{\sigma_W^2}{\sigma_X^2+\sigma_W^2}\mu_X\right)\right)^2}{2\left(\frac{\sigma_X\sigma_W}{\sqrt{\sigma_X^2+\sigma_W^2}}\right)^2}}\right)}{\frac{\sigma_X\sigma_W}{\sqrt{\sigma_X^2+\sigma_W^2}}\sqrt{2\pi}}.
\end{align}
That is, conditioned on $Y=y$, $X$ has Gaussian distribution with mean $\frac{\sigma_X^2}{\sigma_X^2+\sigma_W^2}y+\frac{\sigma_W^2}{\sigma_X^2+\sigma_W^2}\mu_X$, and variance $\frac{\sigma_X^2\sigma_W^2}{\sigma_X^2+\sigma_W^2}$. Thus, 
\begin{align}
    \E{X\given Y} &= \frac{\sigma_X^2}{\sigma_X^2+\sigma_W^2}Y+\frac{\sigma_W^2}{\sigma_X^2+\sigma_W^2}\mu_X \\
    \Var{X\given Y} &= \frac{\sigma_X^2\sigma_W^2}{\sigma_X^2+\sigma_W^2} = \mmse{X\given Y},
\end{align}
and $\E{X\given Y}$ is Gaussian with mean $\mu_X$, variance $\frac{\sigma_X^4}{\sigma_X^2+\sigma_W^2}$. The Fisher information of $\E{X\given Y}$ is the reciprocal of its variance since it is Gaussian, i.e., 
\begin{equation}
    J\left(\E{X\given Y}\right) = \frac{\sigma_X^2+\sigma_W^2}{\sigma_X^4},
\end{equation}
and the main inequality (\ref{mres}) is satisfied with equality: 
\begin{align}
    h(\E{X\given Y}) &= \frac{1}{2}\loge{2\pi e \frac{\sigma_X^4}{\sigma_X^2+\sigma_W^2}}\\&=  2h(X)-h(Y).
\end{align}

}

\section{Main Results}\label{mainresultsection}
The differential entropy of the conditional expectation shows up as a lower/upper bound in certain multi-terminal information theory problems as shown in \cite{eswaran_gastpar_2019}. Thus, it is useful to derive tight bounds on $h(\E{X\given Y})$ to obtain further insights. In general, deriving upper bounds to the differential entropy is not difficult using the maximum entropy argument, while lower bounds are less trivial and they might require extra assumptions on the input distribution such as log-concavity. Using the main result of \cite{lowerbound}, one obtains $h(\E{X\given Y})\ge \frac{1}{2}\,\loge{4 \Var{\E{X\given Y}}}\ge \frac{1}{2}\loge{4\frac{\sigma_X^4}{\sigma_X^2+\sigma_W^2}}$ as long as the density of $\E{X\given Y}$ is log-concave, which might be tedious to check.  
\par For the case of additive Gaussian noise, we present a new lower bound which combines the maximum entropy argument in the conditional setting with the identity (\ref{hV}), and is applicable regardless of the input distribution. 
Our theorem is based on the following lemma, which relates the differential entropy of the conditional expectation to that of the output. 
\begin{lemma}\label{lemma-main}\hspace*{.25em}(Calculation of $h(\E{X\given Y})$):
For the model given in Equation~\eqref{themodel} with $\sigma_W^2 >0,$ the differential entropy of the conditional mean can be written as 
\begin{equation}
    h(\E{X\given Y}) = h(Y) + \E{\loge{\frac{1}{\sigma_W^2}\Var{X\given Y}}}.
    \label{hV}
\end{equation}
\end{lemma}
\begin{proof}
This lemma follows by careful application of several known tools, including Tweedie's formula~\cite{robbins2020empirical} and the Hatsell-Nolte identity~\cite{hatsell_nolte_1971}. A full proof is provided in Appendix~\ref{proofofidentity}. 
\end{proof}
We can now leverage this lemma to the lower bound on the differential entropy of the conditional mean. 
\begin{theorem}\label{thm-Gauss-scalar}:\hspace*{.25em}
Let $X$ be an arbitrary continuous random variable with finite variance and $Y=X+W,$ where $W$ is a zero mean Gaussian with variance $\sigma_W^2$ and independent of $X.$ Then,
\begin{equation}
    h(\E{X\given Y}) \ge 2h(X)-h(Y).
    \label{mainresult}
\end{equation}
Furthermore, equality is achieved if and only if $X$ is Gaussian.
\end{theorem}

\begin{proof}
The proof of this theorem starts from Lemma~\ref{lemma-main} and lower bounds the second summand on the right hand side of Equation~\eqref{hV} by a maximum entropy argument.
Specifically, observe that $h(X\given Y=y)\leq \frac{1}{2}\,\loge{2\pi e\,\Var{X\given Y=y}}$ for every $y$. Taking the expectation of both sides, one obtains $\restarthintedrel$
\begin{IEEEeqnarray}{rCl}
    \E{\loge{\Var{X\given Y}}} &\geq  &2\left(h(X\given Y)-\frac{1}{2}\loge{2\pi e}\right)\\
    &= &2\left(h(X\given Y)-h(W)+\log\sigma_W\right) \nonumber\\
    &=  &2\left(h(X)-h(Y)+\log\sigma_W\right) \label{elogvar3}
\end{IEEEeqnarray}
where (\ref{elogvar3}) follows by the definition of mutual information 
\begin{equation}
    I(X;Y) = h(Y) - h(W) = h(X)-h(X\given Y).
\end{equation}
Combining (\ref{hV}) and (\ref{elogvar3}), the $\log\sigma_W$ term cancels out and the desired result follows.
The uniqueness of the equality condition follows from the uniqueness of the Gaussian distribution as a maximum entropy distribution under a variance constraint.
\end{proof}
To the best of our knowledge, Equation~\eqref{mainresult} is not known in the literature outside of the special equality case when the input is Gaussian \cite{eswaran_gastpar_2019}. As immediate applications, Equation~\eqref{mainresult} can be used to compare the tightness of rate-distortion lower bounds in remote source coding under Gaussian noise, which is explained in Section \ref{mainappsection}; and novel rate loss bounds can be derived as explained in Section \ref{novelratelossbounds}.   

\begin{remark}.
\\
Our main result, Theorem~\ref{thm-Gauss-scalar}, can be expressed in terms of entropy powers as
\begin{subequations}
\begin{equation}
    N(\E{X\given Y})\,N(Y)\geq (N(X))^2.
\end{equation}
There is a pleasing \emph{duality} with a known result about variances,
\begin{equation}
\Var{\E{X\given Y}}\,\sigma_Y^2\geq (\sigma_X^2)^2. \label{sigmaversion}
\end{equation}
\end{subequations}
This last inequality follows directly from the fact that $\mmse{X\given Y}$
cannot be larger than the mean-squared error of the best \emph{linear} estimator,
\begin{equation}
    \mmse{X\given Y}\leq \frac{\sigma_X^2\sigma_W^2}{\sigma_Y^2}, 
    \label{sigmaversion3}
\end{equation}
combined with the fact that
\begin{equation}
    \mmse{X\given Y} + \Var{\E{X\given Y}} = \sigma_X^2, \label{lawoftotalvariance}
\end{equation}
which is the law of total variance.

\end{remark}

\section{Further Properties of the Differential Entropy and Entropy Power of Conditional Mean}\label{furthersection}
\subsection{Upper Bounds}
By the concavity of the logarithm, one can use Jensen's Inequality to derive an upper bound.
\begin{lemma}{\hspace*{.25em}(Upper Bounds of Differential Entropy):}\label{lemma-Upperbounds} For the model given in \ref{themodel}, the differential entropy of the conditional mean can be upper bounded as 
\begin{IEEEeqnarray}{rCl}
    h(\E{X\given Y}) &= &h(Y) + \E{\log{\Var{X\given Y}}}-\log\sigma_W^2\\ 
         &\leq& h(Y) + \log\E{\Var{X\given Y}}-\log\sigma_W^2 \label{ub1}\\
         &=&h(Y) + \log\mmse{X\given Y}-\log\sigma_W^2 \label{ub2}\\
         &\leq &h(Y) + \loge{\frac{\sigma_X^2}{\sigma_X^2 +\sigma_W^2}} \label{ub3}\\
         &\leq &\frac{1}{2} \loge{2\pi e\,\frac{\sigma_X^4}{\sigma_X^2+\sigma_W^2}} \label{ub4}. 
\end{IEEEeqnarray}
\end{lemma}
\begin{proof}
To prove this lemma, we note that (\ref{ub1}) follows from the Jensen's Inequality, (\ref{ub2}) is by definition of $\mathsf{mmse}$, and (\ref{ub3}) and (\ref{ub4}) follow from the maximization of $\mathsf{mmse}$ and differential entropy, respectively.
\end{proof}

Each inequality in Lemma~\ref{lemma-Upperbounds} is satisfied with equality if and only if the input $X$ is Gaussian. That is, for Gaussian inputs, $\Var{X\given Y}$ is constant almost surely, and $\mathsf{mmse}$ and differential entropy are maximized \cite{mmse_guo_2}.
From (\ref{ub4}), it is evident that the maximum differential entropy of the conditional mean is achieved when the input $X$ is Gaussian even though it minimizes the variance of $\E{X\given Y}$, which we record in the following corollary:
\begin{corollary}{\hspace{.25em}(Maximum Entropy):}
For the model of Figure~\ref{observation}, we have
\begin{equation}
    \max_{p_X} h(\E{X\given Y}) = \frac{1}{2} \loge{2\pi e\,\frac{\sigma_X^4}{\sigma_X^2+\sigma_W^2}},
    \label{maximization}
\end{equation}
and the maximum is achieved when $X\sim \mathcal{N}(\mu_X,\sigma_X^2)$.
\end{corollary}

In the sequel, we will sometimes find it convenient to rewrite our upper and lower bounds on the differential entropy of the conditional mean in terms of its entropy power. Namely, Equations~\eqref{mainresult},~\eqref{ub2}, and~\eqref{ub4} can be equivalently expressed as
\begin{align}
\textstyle
    \frac{N^2(X)}{N(Y)}\leq N(\E{X\given Y}) \leq N(Y)\frac{(\mmse{X\given Y})^2}{\sigma_W^4}\leq \frac{\sigma_X^4}{\sigma_X^2+\sigma_W^2}.
\label{remark-six}
\end{align}
\revision{
\begin{remark}. The upper bound (\ref{remark-six}) can also be written as
\begin{equation}
    N(\E{X\given Y})\leq \sigma_X^2-\sigma_W^2\frac{\sigma_X^2}{\sigma_Y^2}.
\end{equation}
There is also a \emph{dual} result about variances, which follows from \cite[Eqn.~16]{eswaran_gastpar_2019} and the law of total variance.
\begin{equation}
    \Var{\E{X\given Y}}\leq \sigma_X^2-\sigma_W^2\frac{N(X)}{N(Y)}.
\end{equation}
\end{remark}}
\subsection{Comparison of The Lower and Upper Bounds}

When the input $X$ is Gaussian, the upper and lower bounds are satisfied with equality. \revision{On the other extreme, as the input distribution approaches to a discrete random variable, the gap between $2h(X)-h(Y)$ and $h(\E{X\given Y})$ increases since the input differential entropy $h(X)$ approaches to $-\infty$. To illustrate, suppose the input is a zero mean Gaussian mixture random variable with two components centered around $-1,\,1$. That is, 
\begin{equation}
    X = 2B-1+\widetilde{X}
\end{equation}
where $B$ and $\widetilde{X}$ are Bernoulli$(\frac{1}{2})$ and Gaussian random variables with mean $0$ and variance $\sigma_X^2-1$, respectively; and $B,\,\widetilde{X},$ and $W$ are independent. We refer to this example as Gaussian mixture input.
In Figure~\ref{figure1}, the bounds are illustrated for the case when the input $X$ follows Gaussian mixture, exponential and uniform distribution. Note that $h(Y)$, $h(\E{X\given Y})$, and $\mathsf{mmse}$ are calculated numerically for exponential and uniform inputs as there are no closed-form expressions. For the Gaussian mixture input, all quantities are calculated numerically.} We also calculated the bounds for other input distributions, e.g., Laplace and Triangular, and observed that there are no visible gaps. Without loss of generality, the noise variance is set to unity in both simulations.

\begin{figure*}[!htb]
    \centering
    \includegraphics[width=0.3215\textwidth]{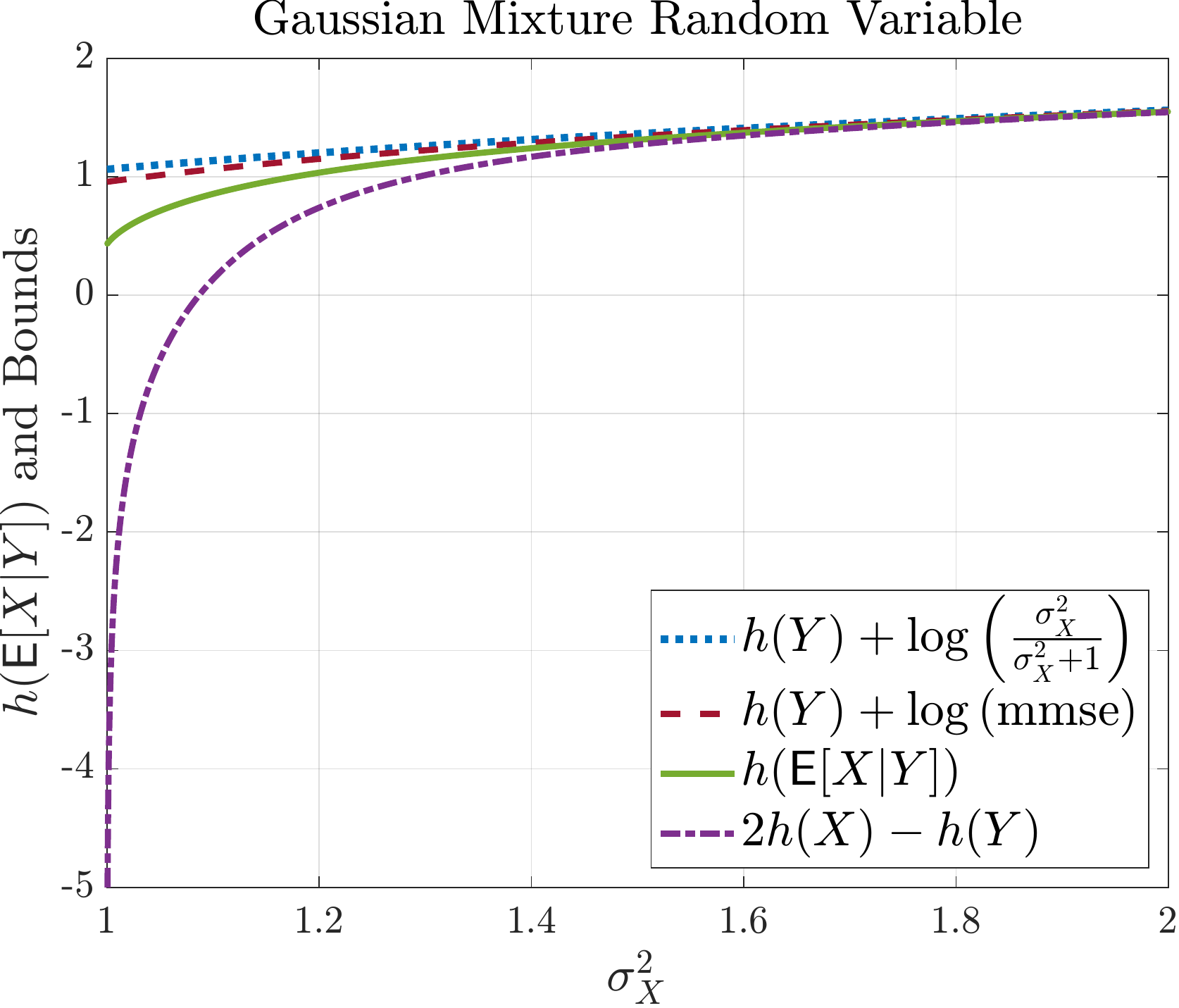}\hspace*{.15cm}
    \includegraphics[width=0.325\textwidth]{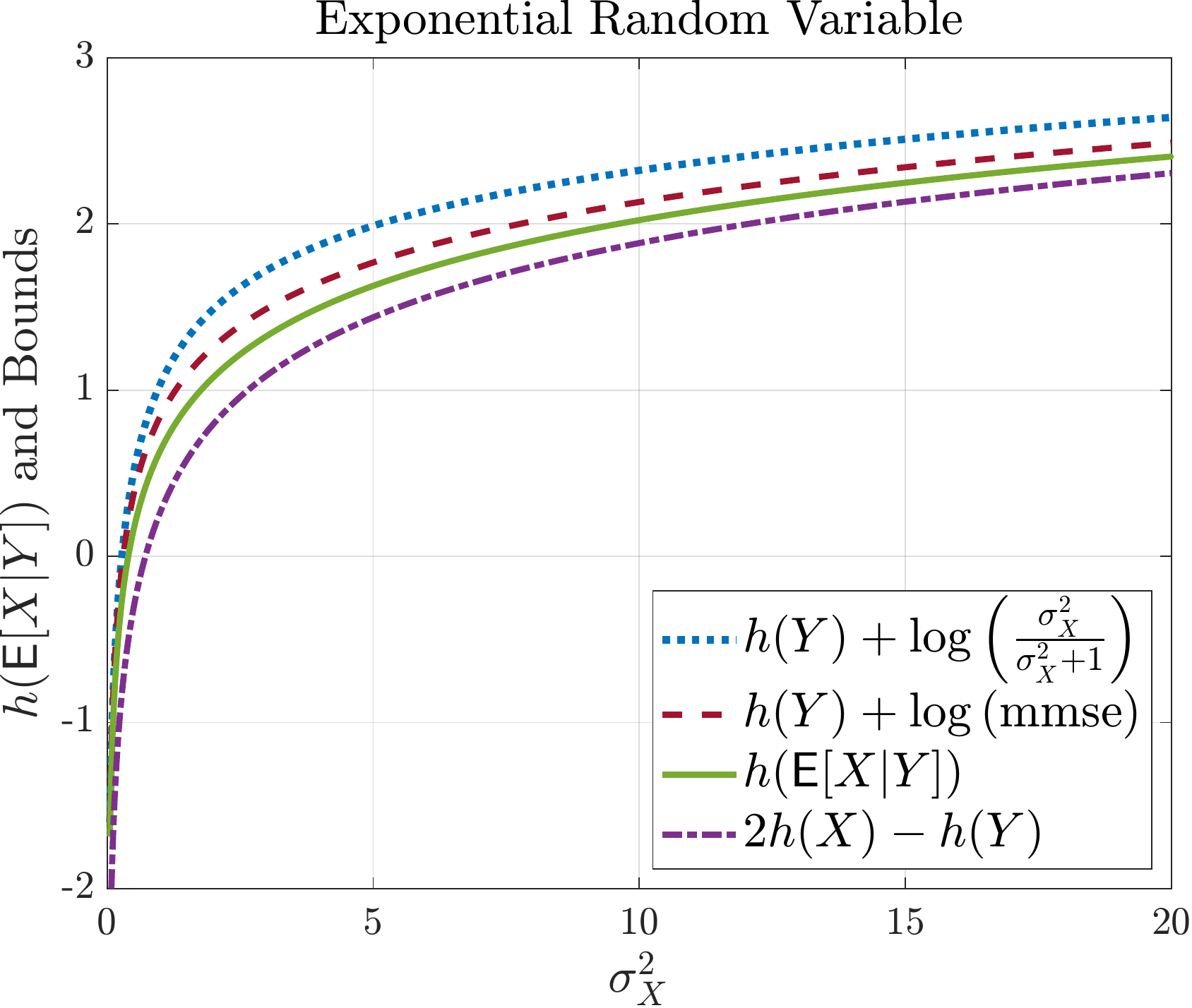}\hspace*{.15cm}
    \includegraphics[width=0.325\textwidth]{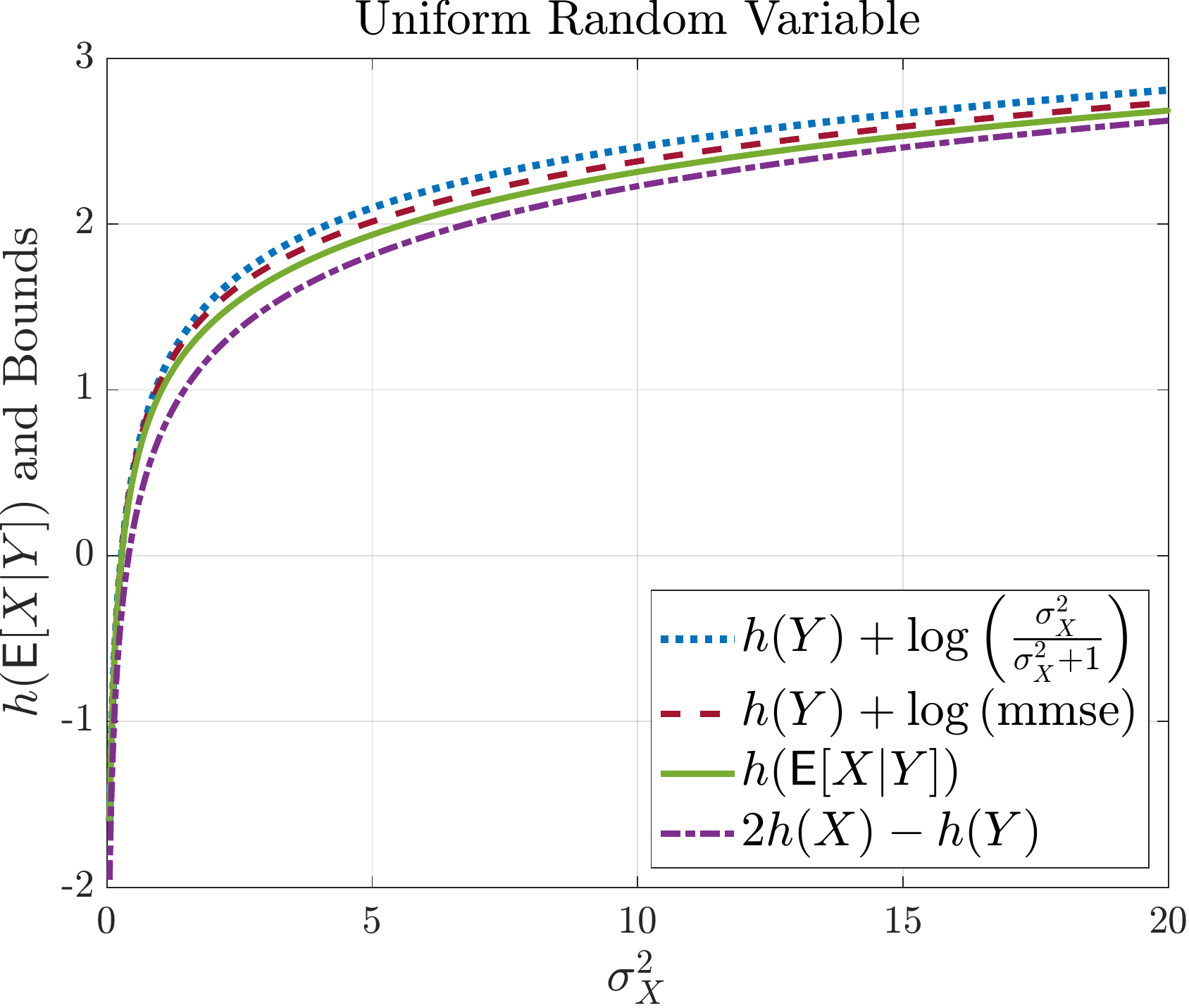}
    \caption{\revision{$h(\E{X\given Y})$ and the bounds (\ref{mainresult}), (\ref{ub2}), (\ref{ub3}) for Gaussian mixture, exponential and uniform inputs.}}
    \label{figure1}
\end{figure*}
\par Observe from (\ref{figure1}) that the bounds are tighter for input distributions close to Gaussian in terms of Kullback-Leibler (KL) divergence, and both the upper and lower bounds become tighter as the input variance increases.  

\subsection{Input Variance Asymptotics}
We fix the noise variance $\sigma_W^2$, and study low and high $\sigma_X^2$ asymptotics of $h(\E{X\given Y})$ and $N(\E{X\given Y})$. 
 \begin{lemma}{\hspace*{.25em}(Low $\sigma_X^2$ Asymptotics):}
As $\sigma_X^2\rightarrow 0^+$
\begin{subequations}
 \begin{align}
     h(\E{X\given Y}) &\sim \log\sigma_X^2 \label{lowsnr1}\\
     N(\E{X\given Y}) &\sim \sigma_X^4 \label{lowsnr2}
 \end{align}
 \end{subequations}
 i.e., $\lim_{\sigma_X^2\to 0^+} \frac{h(\E{X\given Y})}{\log\sigma_X^2} = 1$ and $\lim_{\sigma_X^2\to 0^+} \frac{N(\E{X\given Y})}{\sigma_X^4}=1$.
 \end{lemma}
 \begin{proof}
 Observe that as $\sigma_X^2$ approaches $0^+$, since $Y$ converges to $\mu_X+W$ almost surely, the first term in the right hand side of (\ref{hV}) approaches to the differential entropy of the noise, i.e., $h(Y)= - \int p_Y(y) \log p_Y(y)dy\rightarrow - \int p_W(y) \log p_W(y)dy = h(W)$. For the second term, it is easy to check that $\Var{X\vert Y}$ converges to $\Var{X}=\sigma_X^2$ almost surely, which implies that $\E{\log\Var{X\vert Y}}\rightarrow \log \sigma_X^2$. Combining these, we obtain (\ref{lowsnr1}), which is equivalent to (\ref{lowsnr2}).
 \end{proof}

\begin{remark}{\hspace*{.25em}(High $\sigma_X^2$ Asymptotics):}
As $\sigma_X^2\rightarrow \infty$
\begin{subequations}
 \begin{align}
     h(\E{X\given Y}) &= \mathcal{O}\left(\log \sigma_X\right) \label{highsnr1}\\
     N(\E{X\given Y}) &= \mathcal{O}\left(\sigma_X^2\right). \label{highsnr2}
 \end{align}
 \end{subequations}
 \end{remark}
 \begin{proof}
 Proof directly follows from (\ref{ub4}), i.e.,  $h(\E{X\given Y}) \leq \frac{1}{2}\loge{2\pi e}+
 \frac{1}{2}\loge{\frac{\sigma_X^2}{\sigma_X^2 +\sigma_W^2}}+\frac{1}{2}\log\sigma_X^2 $ and $N(\E{X\given Y}) \leq \frac{\sigma_X^4}{\sigma_X^2+\sigma_W^2}$
 \end{proof} 
 
\subsection{Taylor Series Expansion}
In this subsection, we discuss Taylor series approximation to $\E{\loge{\Var{X\given Y}}}$. Assume that all moments of $\Var{X\given Y}$ are finite. Denote the $k^{th}$ central moment of $\Var{X\given Y}$ by $c_k$, i.e., $\E{(\Var{X\given Y}-\E{\Var{X\given Y}})^k} = c_k$ and $\mmse{X\given Y}$ by $\mathsf{mmse}$. 
Since $\loge{\cdot}$ is sufficiently differentiable, we have
\begin{IEEEeqnarray}{rCl}
\E{\log{\Var{X\given Y}}} &=& \log\mathsf{mmse}-\frac{c_2}{2\mathsf{mmse}^2} +  \frac{c_3}{3\,\xi_S^3}
\end{IEEEeqnarray}
where $\xi$ satisfies
\begin{equation}
    \left|\xi-\E{\Var{X\given Y}}\right|<\left|\Var{X\given Y}-\E{\Var{X\given Y}}\right|.
\end{equation}

Using (\ref{hV}), Taylor expansion of $h(\E{X\given Y})$ follows.
\begin{align}
    h(\E{X\given Y}) = h(Y) + \loge{\frac{\mathsf{mmse}}{\sigma_W^2}}-\frac{c_2}{2\mathsf{mmse}^2} +  \frac{c_3}{3\,\xi_S^3}.\label{taylor}
\end{align}
Observe that neglecting the last two terms in (\ref{taylor}), we recover the upper bound in (\ref{ub2}), and neglecting only the last term, we obtain an approximation to $h(\E{X\given Y})$.
\begin{align}
    h(\E{X\given Y}) \approx h(Y) + \loge{\frac{\mathsf{mmse}}{\sigma_W^2}}-\frac{\Var{\Var{X\given Y}}}{2\mathsf{mmse}^2}.
\end{align}

\subsection{Connection to Fisher Information}
A lower bound for the Fisher Information of $V$ directly follows from Stam's Inequality \cite{stam_fisher}.
\begin{equation}
J(\E{X\given Y})\geq \frac{1}{N(\E{X\given Y})}.
\label{pr1}
\end{equation}

We can further lower bound this by applying Equation~\eqref{remark-six}, obtaining
\begin{align}
J(\E{X\given Y}) &\geq \frac{1}{N(Y)}\frac{\sigma_W^4}{\left(\mmse{X\given Y}\right)^2} \\&\geq \frac{\sigma_X^2+\sigma_W^2}{\sigma_X^4},
\end{align}
and noting that equality is achieved when the input is Gaussian.
This lower bound on $J(\E{X\given Y})$ is of interest since it does not involve the differential entropy of $\E{X\given Y}.$

\revision{
\subsection{Comparison with Costa's Entropy Power Inequality}\label{comparewithepi}
\par Define $Y_\alpha \triangleq X+\alpha W$, where $\alpha\in [0,1]$, and $X,W$ satisfy the same assumptions as the main model (\ref{themodel}). Observe from (\ref{elogvar3}) that the lower bound  (\ref{mainresult}), applied to the modified model $Y_{\alpha}$ is equivalent to 
\begin{align}
N(Y_\alpha) &\geq \alpha^2\,\sigma_W^2\, N(X)\, \exp\left(-\E{\log \Var{X\given Y_\alpha}}\right). \label{epilike}
\end{align}

Costa's entropy power inequality (EPI) \cite{costa} gives a lower bound on $N(Y_\alpha)$ for $\alpha\in [0,1]$ in terms of $N(X)$ and $N(Y_1)$. That is,
\begin{equation}
    N(Y_\alpha) \geq (1-\alpha^2)N(X)+\alpha^2N(Y_1).\label{costaepi}
\end{equation}
It is tempting to compare (\ref{epilike}) with the Costa's EPI as they both provide lower bounds on the entropy power of the output $Y_\alpha$. 
In Fig. \ref{figure2}, we compare these bounds when the input $X$ is a uniform random variable with zero mean and variance $\sigma_X^2.$ Moreover, we set $\sigma_W=1$ and consider $\alpha\in\{\frac{2}{5},\frac{2}{3}\}$. The figure illustrates the gap to the main result, which is 
\begin{subequations}
\begin{equation}
    N(Y_\alpha)-\alpha^2\, N(X)\, \exp\left(-\E{\log \Var{X\given Y_\alpha}}\right),\label{gaptomainresult}
\end{equation} 
and the gap to Costa's EPI, which is 
\begin{equation}
    N(Y_\alpha)-\left((1-\alpha^2)N(X)+\alpha^2N(Y_1)\right).\label{gaptocostasepi}
\end{equation} \end{subequations}
As expected, (\ref{epilike}) performs better for small values of $\alpha$.}

\begin{figure}[!htb]
    \centering
    \includegraphics[width = \linewidth]{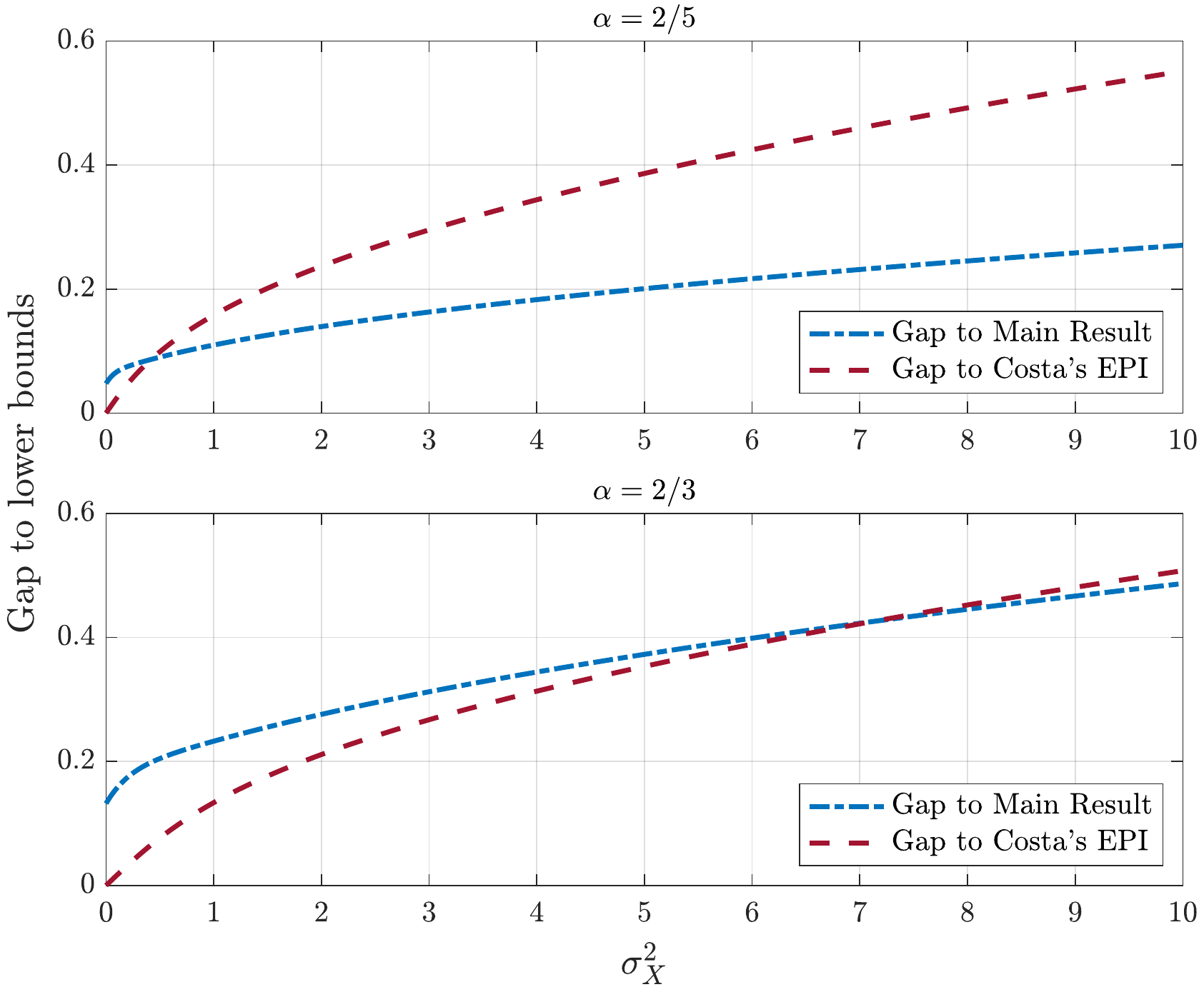}
    \caption{Comparison of the tightness of lower bounds for $N(Y_{\alpha})$: Gap to Main Result and Gap to Costa's EPI refer to (\ref{gaptomainresult}) and (\ref{gaptocostasepi}), respectively.}
    \label{figure2}
\end{figure}

\section{Applications}\label{applicationsection}

\subsection{Lower Bounds of the Rate Distortion Function in the Remote Source Coding Problem}\label{sec-remote}
\label{mainappsection}
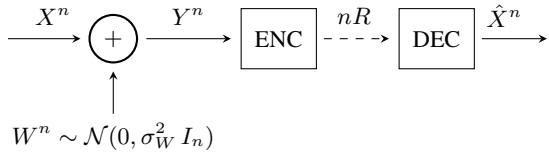
\begin{figure}[H]
\begin{center}
\small
\begin{tikzpicture}[scale=2,shorten >=1pt, auto, node distance=1cm,
   node_style/.style={scale=1,circle,draw=black,thick},
   edge_style/.style={draw=black,dashed}]

    \node [fill=none] at (-.4,0.15) (nodeS) {$X^n$};
    \node [fill=none] at (0,-0.65) (nodeS) {$W^n \sim \mathcal{N}(0,\sigma_W^2\,I_n)$};
    \node [fill=none] at (0.5,0.15) (nodeS) {$Y^n$};
    \node[rectangle,draw,minimum width = 1cm, 
    minimum height = .8cm] (r) at (1.1,0) {ENC};
    \node[node_style] (v1) at (0,0) {$+$};
    \node [fill=none] at (1.6,0.15) (nodeS) {$nR$};
    \node[rectangle,draw,minimum width = 1cm, 
    minimum height = .8cm] (r) at (2.15,0) {DEC};
    \node [fill=none] at (2.6,0.15) (nodeS) {$\hat{X}^n$};
   
    \draw [-stealth](-0.7,0) -- (-0.2,0);
    \draw [-stealth](0,-0.5) -- (0,-0.2);
    \draw [-stealth](0.22,0) -- (0.8,0);
    \draw [dashed,-stealth](1.4,0) -- (1.85,0);
    \draw [-stealth](2.45,0) -- (2.9,0);
    \end{tikzpicture}
\end{center}
  \caption{The AWGN remote source coding problem.}\label{Fig-AWGNremoteSC}
\end{figure}

An important application of Inequality (\ref{mainresult}) can be found in the remote source coding problem.
Specifically, consider the source coding problem illustrated in Figure~\ref{Fig-AWGNremoteSC}: An encoder observes the underlying source $X$ subject to additive white Gaussian noise $W.$ The noisy observation is $Y$ and can be encoded using $R$ bits per sample. The decoder produces a reconstruction $\hat{X}$ to within the smallest possible mean-squared error. For a formal problem statement, we refer to~\cite{eswaran_gastpar_2019}.
The smallest possible rate to attain a target distortion $D$ is referred to as the {\it remote rate-distortion function,} denoted as $R_X^R(D).$
For the case where the underlying source $X$, not necessarily Gaussian, has finite differential entropy,~\cite{eswaran_gastpar_2019} discusses two different lower bounds for the remote rate-distortion function, namely
\begin{IEEEeqnarray}{rCl}
R_X^R(D) &\geq& \frac{1}{2}\log^{+} \frac{N(\E{X\given Y})}{D} \nonumber \\
&& {}+\frac{1}{2}\log^{+} \frac{N(Y)}{N(Y)-\frac{N(X)}{D}N(W)}, \label{lb1}
\end{IEEEeqnarray}
and
\begin{IEEEeqnarray}{rCl}
R_X^R(D)&\geq& \frac{1}{2}\log^{+} \frac{N(X)}{D}\nonumber\\&& {}+ \frac{1}{2}\log^{+} \frac{N(X)}{N(Y)-\frac{N(X)}{D}\sigma_W^2}
\label{lb2}
\end{IEEEeqnarray}
where $D>\E{(X-\E{X\given Y})^2}$ and $\log^{+}x=\max\{0,\,\log x\}$. At the time of their writing, a comparison between the right-hand sides (RHS) of (\ref{lb1}) and (\ref{lb2}) was not available. We prove that the RHS of (\ref{lb1}) is always greater than or equal to the RHS of (\ref{lb2}).
\begin{proposition}{\hspace*{.25em}(A comparison of lower bounds for the remote rate-distortion function):}
\begin{multline}
    \frac{1}{2}\log^{+} \frac{N(\E{X\given Y})}{D}+ \frac{1}{2}\log^{+} \frac{N(Y)}{N(Y)-\frac{N(X)}{D}N(W)}  \\
    \geq \frac{1}{2}\log^{+} \frac{N(X)}{D}+ \frac{1}{2}\log^{+} \frac{N(X)}{N(Y)-\frac{N(X)}{D}\sigma_W^2}.
\label{mainapp}
\end{multline}
Furthermore, equality is achieved if and only if $X$ is Gaussian.
\end{proposition}

\begin{proof}
First, note that since $W$ is Gaussian, we have that $N(W) = \sigma_W^2.$ Moreover, by the entropy power inequality, 
\begin{equation}
N(Y)\geq N(X)+N(W),
\end{equation}
and clearly, $N(Y)> N(Y)- \frac{N(X)}{D}N(W).$
Combining these observations, we conclude that the second term in the RHS of (\ref{lb1}) is always greater than or equal to the second term in (\ref{lb2})

To compare the first terms in the inequalities, observe that if $D\geq N(X)\geq N(\E{X\given Y})$, both $\log^{+}\frac{N(\E{X\given Y})}{D}$ and $\log^{+}\frac{N(X)}{D}$ are equal to 0, so the lower bound in (\ref{lb1}) is greater. If $N(X)\geq N(\E{X\given Y}) > D$, since $N(\E{X\given Y})N(Y)\geq (N(X))^2$ from Proposition (\ref{mainresult}), it is again guaranteed that the lower bound in (\ref{lb1}) is greater. Finally, the last possible and a little more complicated case is $N(X)> D \geq N(\E{X\given Y})$. In this case, the comparison of the lower bounds seen in (\ref{lb1}) and (\ref{lb2}) is equivalent to the comparison of $DN(Y)$ and $(N(X))^2$. Taking the logarithm of both sides, we thus have to compare
\begin{equation}
2h(Y) + \log 2\pi eD \mathop{\lessgtr}_{(\ref{lb1})}^{(\ref{lb2})} 4h(X).
\label{comp_eq}
\end{equation}

Note that we have assumed $D \geq N(\E{X\given Y})$ for this case at the beginning. This is equivalent to the following inequality. 
\begin{equation}
\log 2\pi eD \geq 2h(\E{X\given Y}).
\label{cond_eq}
\end{equation}
Combining Theorem~\ref{thm-Gauss-scalar} (that is, Equation~\eqref{mainresult}) with Equations~\eqref{comp_eq} and~\eqref{cond_eq}, one obtains
\begin{equation}
2h(Y) + \log 2\pi eD \geq 2h(Y) + 2h(\E{X\given Y}) \geq 4h(X).
\label{cond_eq2}
\end{equation}
Therefore, we conclude that the lower bound in (\ref{lb1}) is always greater than or equal to the lower bound in (\ref{lb2}).
\end{proof}

\subsection{Novel Rate-Loss Bounds for The CEO Problem}
\label{novelratelossbounds}

In this section, we apply Theorem~\ref{thm-Gauss-scalar} to the so-called CEO problem.
In this problem, a single underlying source $X$ is observed by $M$ encoders. Each encoder receives a noisy version of the source $X,$ denoted as $Y_i=X+W_i,$ for $i=1, 2, \ldots, M.$ In our consideration, the noises $W_i$ are assumed to be zero-mean Gaussian, independent of each other, and of variance $\sigma_W^2.$
Each encoder compresses its observation using $R_i$ bits. All $M$ compressed representations are given to a single central decoder whose goal is to produce a reconstruction of the underlying source $X$ to with mean-squared error $D.$ The smallest possible sum-rate required to attain a distortion $D$ is denoted by $R_X^{\mathrm{CEO}}(D).$
We precisely follow the exact problem statement and notation used in~\cite{eswaran_gastpar_2019}.

The {\it rate loss} in the CEO problem denotes the difference between $R_X^{\mathrm{CEO}}(D)$ and the much smaller rate that would be required if all encoders were to cooperate fully, {\it i.e.,} the rate required by a single encoder having access to all $M$ noisy source observations.
Evidently, if the encoders are allowed to cooperate fully, then the problem is exactly the remote rate-distortion problem discussed in Section~\ref{sec-remote} above, but with reduced variance $\sigma_W^2/M$.
We denote the corresponding rate by $R_{X}^{\mathrm{R}}(D),$ and the rate loss by
\begin{equation}
    L(D) \triangleq R_X^{\mathrm{CEO}}(D)-R_{X}^{\mathrm{R}}(D).
\end{equation} 
In this section, we establish a novel bound on this rate loss.

To develop our results we will use the auxiliary notation
\begin{align}
    Y(M) &= \frac{1}{M} \sum_{i=1}^M Y_i.
\end{align}

A lower bound for $L(D)$ is presented in~\cite{eswaran_gastpar_2019}.
For
\begin{align}
\sigma_X^2\dfrac{\sfrac{\sigma_W^2}{M}}{\sigma_X^2+\sfrac{\sigma_W^2}{M}}<D<N(X)\dfrac{\sfrac{\sigma_W^2}{M}}{N(Y(M))-N(X)},
\end{align}
the lower bound on the rate loss establishes that
\begin{multline}
    L(D)\geq \frac{M}{2}\loge {\frac{1}{\frac{N(Y(M))}{N(X)}-\frac{\sigma_W^2}{M}\frac{1}{D}}}\\
    -\frac{1}{2}\loge{\frac{\sigma_X^2}{N(X)}\,\frac{1}{1-\frac{\sigma_W^2}{M}\left(\frac{1}{D}-\frac{1}{\sigma_X^2}\right)}}. \label{lowerbnd1}
\end{multline}
For $M\uparrow \infty$ and under some regularity conditions further discussed in Appendix \ref{kappacalc},
the bound becomes, for $0<D<\frac{1}{J(X)}$
\begin{equation}
    L(D) \geq  \frac{\sigma_W^2}{2} \left(\frac{1}{D}-{J(X)}\right)-\frac{1}{2}\log \frac{\sigma_X^2}{N(X)}.\label{Eqn-rateloss-lowerbound}
\end{equation}
The novel bound presented here is an {\it upper} bound on the rate loss, developed in the following subsections.

\subsubsection{Cooperation Bound}
The first ingredient of the novel upper bound on the rate loss is an improved lower bound on $R_{X}^{\mathrm{R}}(D).$
 To this end, we will utilize both $N(\E{X\given Y(M)})$ and $\mmse{X\given Y(M) }$, i.e., for all $D> \mmse{X\given Y(M)}$
\begin{align}
    R_{X}^{\mathrm{R}}(D) &\geq \frac{1}{2}\log^+ \frac{N(\E{X\given Y(M)})}{D-\mmse{X\given Y(M)}}. \label{tight}
\end{align}

One can weaken (\ref{tight}) to omit the calculation of $\mmse{X\given Y(M)}$. In that case, one obtains for all $D>N(X)\,\sigma_W^2/\left(M\,N(Y(M))\right)$,
\begin{IEEEeqnarray}{rCl}
    R_{X}^{\mathrm{R}}(D) &\geq& \frac{1}{2}\log^{+}\frac{N(\E{X\given Y(M)})}{D}\nonumber\\
    && {}+\frac{1}{2}\log^{+}\frac{M\, N\left(Y(M)\right)}{M\, N\left(Y(M)\right)-\frac{N(X)}{D}N(W)} \\
    &=& \begin{cases}
    \frac{1}{2} \log \frac{N(\E{X\given Y(M)})}{D - \frac{\sigma_W^2}{M} \frac{N(X)}{N(Y(M))}} &\small{\text{ if }} \scriptstyle{D<N(\E{X\given Y(M)})}\\[1.5em]
    \frac{1}{2} \log\frac{1}{1-\frac{\sigma_W^2}{M}\, \frac{N(X)}{D\,N(Y(M))}} &\small{\text{ otherwise }}  \label{rcoop}
    \end{cases}
\end{IEEEeqnarray}
where $Y(M) = X+\sum_{i=1}^M\,W_i/M$. As we have shown in Section~\ref{sec-remote}, this bound is tighter than the other lower bound in \cite{eswaran_gastpar_2019} for any finite\footnote{Observe that as $M\uparrow \infty$, the second term vanishes, and the bound becomes $R_{X}^{\mathrm{R}}(D)\geq \frac{1}{2}\log^+ \frac{N(X)}{D}$ as expected. This is also true for the other lower bound.} $M$.  
\subsubsection{Novel Rate Loss Upper Bound}
In order to upper bound the rate loss $L(D),$ we utilize the upper bound on the CEO sum-rate distortion by Eswaran and Gastpar \cite{eswaran_gastpar_2019}, which states that for $D>\sigma_X^2\sigma_W^2/\left(M\,\sigma_{Y(M)}^2\right)$, 
\begin{IEEEeqnarray}{rCl}
    R_X^{\mathrm{CEO}}(D) &\leq& \frac{1}{2}\log^+\frac{\sigma_X^2}{D}\nonumber\\&&\quad+\frac{M}{2}\log^+ \frac{M\,\sigma_X^2}{M \sigma_{Y(M)}^2-\frac{\sigma_X^2}{D}\sigma_W^2}\\
    &=& \frac{1}{2}\log^+\frac{\sigma_X^2}{D}\nonumber\\&&\quad+\frac{M}{2}\log^+\frac{1}{1+\frac{\sigma_W^2}{M}\left(\frac{1}{\sigma_X^2}-\frac{1}{D}\right)}\\
    &=& \begin{cases}\frac{1}{2}\log\frac{\sigma_X^2}{D}\\
     \quad+\frac{M}{2}\log\frac{1}{1+\frac{\sigma_W^2}{M}\left(\frac{1}{\sigma_X^2}-\frac{1}{D}\right)} & \small\text{if } \scriptstyle{D<\sigma_X^2}\\
    0 & \small\text{otherwise} \label{rceo}
    \end{cases}.
\end{IEEEeqnarray}
Since $N(\E{X\given Y(M)})\leq \Var{\E{X\given Y(M)}}\leq \sigma_X^2$ and $\sigma_X^2\sigma_W^2/\left(M\,\sigma_{Y(M)}^2\right)\geq N(X)\,N(W)/\left(M\,N(Y(M))\right)$, we have two region of interests as we subtract (\ref{rcoop}) from (\ref{rceo}) to obtain the new upper bound.
\begin{theorem}.\label{firstthm}
For $D>\sigma_X^2\sigma_W^2/\left(M\,\sigma_{Y(M)}^2\right)$, we have the inequality seen in (\ref{eq-thm9}).
\begin{IEEEeqnarray}{rCl}
    L(D) &\leq&
    \begin{cases}
    \frac{1}{2}\log\left( \frac{\sigma_X^2}{N(\E{X\given Y(M)})}\left(1-\frac{\sigma_W^2}{M\,D}\frac{N(X)}{N(Y(M))}\right)\right)\\\quad+\frac{M}{2}\log\frac{1}{1+\frac{\sigma_W^2}{M}\left(\frac{1}{\sigma_X^2}-\frac{1}{D}\right)}
    & \hspace*{-.2cm}\small\textnormal{if }\; C_1\\
    \frac{1}{2}\log\left( \frac{\sigma_X^2}{D}\left(1-\frac{\sigma_W^2}{M\,D}\frac{N(X)}{N(Y(M))}\right)\right)\\\quad+\frac{M}{2}\log\frac{1}{1+\frac{\sigma_W^2}{M}\left(\frac{1}{\sigma_X^2}-\frac{1}{D}\right)} &
    \hspace*{-.2cm}\small\textnormal{if }\; C_2
    \end{cases}
    \label{eq-thm9}
\end{IEEEeqnarray}
where $C_1$ is the condition $D<N(\E{X\given Y(M)})$ and $C_2$ is the condition $N(\E{X\given Y(M)})\leq D<\sigma_X^2$.
\end{theorem}

As the number of agents approaches infinity, (\ref{eq-thm9}) simplifies to the following. 
\begin{corollary}.\label{cor1}
As $M\uparrow \infty$, the upper bound on the loss becomes
\begin{align}\small
    L(D) \leq
    \begin{cases}
    \frac{1}{2}\log\frac{\sigma_X^2}{N(X)}+\frac{1}{2} \sigma _W^2 \left(\frac{1}{D}-\frac{1}{\sigma _X^2}\right) & \text{if }\, D<N(X)\\
    \frac{1}{2}\log\frac{\sigma_X^2}{D}+\frac{1}{2} \sigma _W^2 \left(\frac{1}{D}-\frac{1}{\sigma _X^2}\right) & \text{if }\, N(X)\leq D<\sigma_X^2
    \end{cases}. \label{newBound}
\end{align}
\end{corollary}

It is also possible to use (\ref{tight}) to obtain a tighter upper bound on the rate loss \revision{for arbitrary $M$}. 
\begin{theorem}.\label{thmtight}
For $\mmse{X\given Y(M)}<D<\mmse{X\given Y(M)}+N(\E{X\given Y(M)})$, the rate loss is upper bounded as 
\begin{IEEEeqnarray}{rCl}
    L(D) &\leq& \frac{1}{2} \log 2\pi e\sigma_X^2 - h(\E{X\given Y(M)})\nonumber\\
    && {}+\frac{1}{2}\log \left(1-\frac{\mmse{X\given Y(M)}}{D}\right)\nonumber\\ && {}+\frac{M}{2}\log\frac{1}{1-\frac{\sigma_W^2}{M}\left(\frac{1}{D}-\frac{1}{\sigma_X^2}\right)}\\
    &=&\frac{1}{2}\log \frac{\sigma_X^2}{N(\E{X\given Y(M)})}\nonumber\\
    && {}+\frac{1}{2}\log \left(1-\frac{\mmse{X\given Y(M)}}{D}\right) \nonumber\\
    && {}+\frac{M}{2}\log\frac{1}{1-\frac{\sigma_W^2}{M}\left(\frac{1}{D}-\frac{1}{\sigma_X^2}\right)}, \label{tightub}
\end{IEEEeqnarray}
and for $\mmse{X\given Y(M)}+N(\E{X\given Y(M)})<D<\sigma_X^2$, we have 
\begin{align}
    L(D)\leq \frac{1}{2}\log\frac{\sigma_X^2}{D}+\frac{M}{2}\log\frac{1}{1-\frac{\sigma_W^2}{M}\left(\frac{1}{D}-\frac{1}{\sigma_X^2}\right)}.
\end{align}
\end{theorem}

\begin{remark}.
    Note that (\ref{tightub}) is minimized for Gaussian inputs since both $N(\E{X\given Y(M)})$ and $\mmse{X\given Y(M)}$ is maximized in that case. Furthermore, the theorem simplifies to 
    \begin{IEEEeqnarray}{rCl}
        L(D)&\leq& \frac{1}{2} \log \left(1-\frac{\sigma_W^2}{M}\left(\frac{1}{D}-\frac{1}{\sigma_X^2}\right)\right)\nonumber\\
        && {}+\frac{M}{2}\log\frac{1}{1-\frac{\sigma_W^2}{M}\left(\frac{1}{D}-\frac{1}{\sigma_X^2}\right)}\\
        &=& \frac{M-1}{2}\log\frac{1}{1-\frac{\sigma_W^2}{M}\left(\frac{1}{D}-\frac{1}{\sigma_X^2}\right)}
    \end{IEEEeqnarray}
    for any $\frac{\sigma_W^2 \sigma_X^2}{M \sigma_X^2+\sigma_W^2}<D<\sigma_X^2$. The exact loss for the Gaussian input is well-known \cite{eswaran_gastpar_2019}.
    \begin{align}
    L_\mathcal{N}(D)&=\frac{M-1}{2} \log \left(\frac{D}{\frac{\sigma_{X}^{2}+\sigma_{W}^{2} / M}{\sigma_{X}^{2}} D-\frac{\sigma_{W}^{2}}{M}}\right)\\  
    &= \frac{M-1}{2}\log\frac{1}{1-\frac{\sigma_W^2}{M}\left(\frac{1}{D}-\frac{1}{\sigma_X^2}\right)}.
    \end{align}
Hence, the new upper bound is tight for Gaussian inputs, irrespective of $M$.  
\end{remark}

\begin{remark}.
\begin{itemize}
    \item As $M\uparrow \infty$, Theorem \ref{firstthm} and \ref{thmtight} yield the same asymptotics, i.e., Corollary \ref{cor1}. Yet, Theorem \ref{thmtight} is guaranteed to be tighter in the finite r\'egime. 
    \item It is important to note that the Gaussian input maximizes \revision{the lower bound} (\ref{lowerbnd1}), whereas it minimizes the upper bound (\ref{tightub}). Hence, the bounds are tight for the inputs that are close to Gaussian distribution.  
\end{itemize}
\end{remark}
\subsubsection{Previous Rate Loss Upper Bound}
We also note that the following upper bound on the rate loss $L(D)$ appears in~\cite{dragotti2009distributed}. 
\begin{IEEEeqnarray}{rCl}
L(D) &\leq  &\frac{M-1}{2} \log \frac{1}{1-\frac{\sigma_{W}^{2}}{M}\left(\frac{1}{D}-\frac{1}{\sigma_X^2}\right)} 
\\
&& {}+\frac{1}{2} \log \scriptstyle\left(1+\frac{\left(\sigma_{X}^{2}-D\right)\left(M \cdot\left(D+2 \sqrt{D \cdot \sigma_{W}^{2}}\right)+\sigma_{W}^{2}\right)}{D \cdot\left(M \sigma_{X}^{2}+\sigma_{W}^{2}\right)-\sigma_{W}^{2} \sigma_{X}^{2}}\right)\nonumber\\
&= &\frac{M-1}{2} \log \frac{1}{1-\frac{\sigma_{W}^{2}}{M}\left(\frac{1}{D}-\frac{1}{\sigma_X^2}\right)}\\
&& {}+\frac{1}{2} \loge{\scriptstyle 1+\left(\frac{1}{D}-\frac{1}{\sigma_X^2}\right)\left(\frac{D+2\sqrt{D}\sigma_W+\frac{\sigma_W^2}{M}}{1-\frac{\sigma_W^2}{M}\left(\frac{1}{D}-\frac{1}{\sigma_X^2}\right)}\right)}, \label{prevupperbound}
\end{IEEEeqnarray}

and as $M\uparrow \infty$, we have 
\begin{multline}
L(D) \leq \frac{\sigma_{W}^{2}}{2}\left(\frac{1}{D}-\frac{1}{\sigma_{X}^{2}}\right)\\
+\frac{1}{2} \log \left(1+\left(\frac{1}{D}-\frac{1}{\sigma_{X}^{2}}\right)\left(D+2 \sqrt{D \sigma_{W}^{2}}\right)\right). \label{prevBound}
\end{multline}

\subsubsection{Comparison of the Rate Loss Bounds}

Comparing (\ref{tightub}) and (\ref{prevupperbound}) for any input distributions is tedious. For Gaussian inputs, we proved that the former is tighter than the latter. For other distributions such as Laplace, Exponential, Uniform we present numerical results in Fig.~\ref{fig1}, and Fig.~\ref{fig2}. Note that the new upper bound is tighter than the general upper bound and its dependence on the number of users $M$ is more prominent at low distortions. \revision{For a detailed analysis on the rate loss in the AWGN CEO problem, we refer to \cite{cisspaper}.}

\begin{figure}[!htb]
    \centering
    \includegraphics[width = \linewidth]{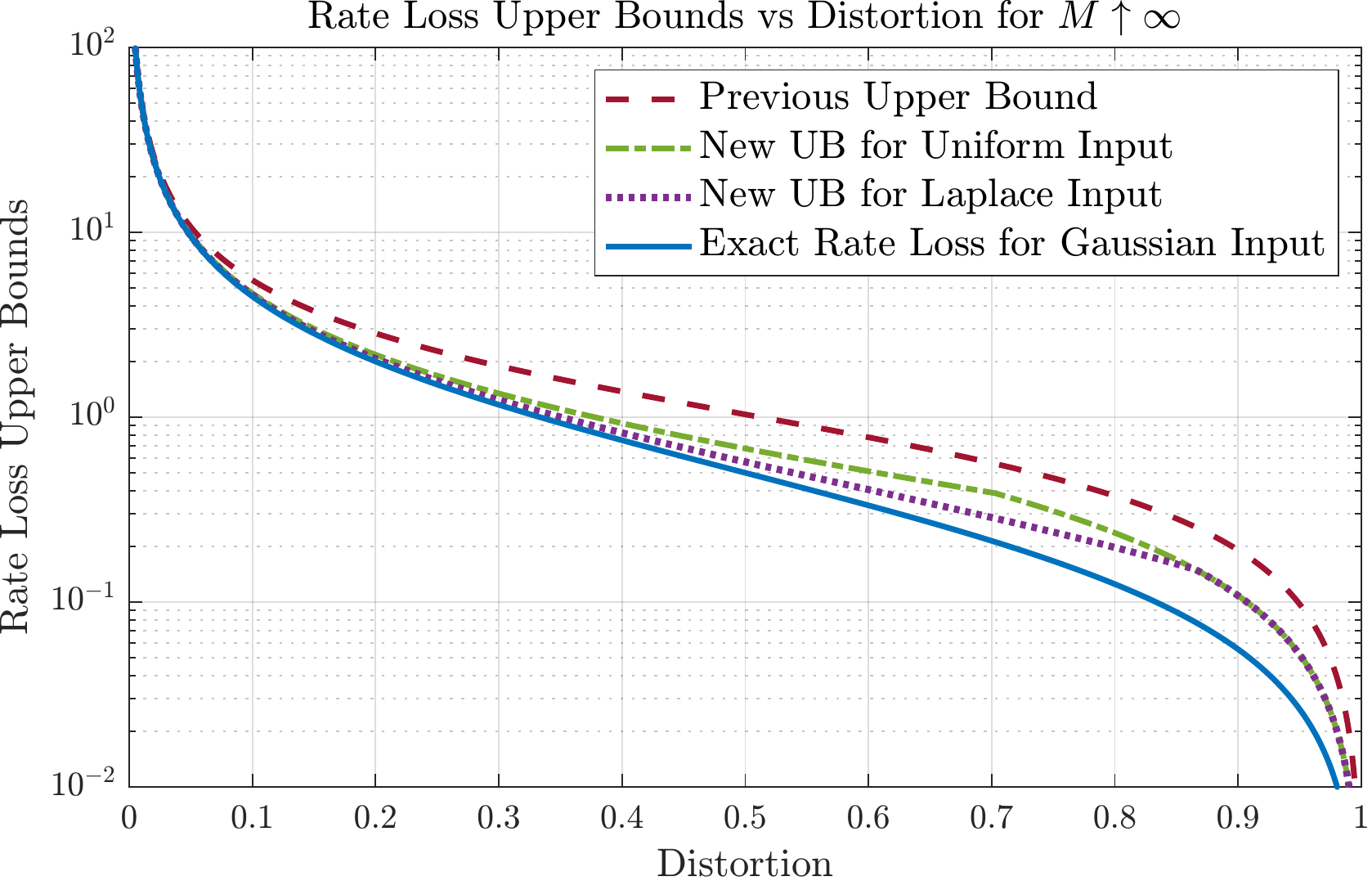}
    \caption{Comparison of Upper Bounds for $\sigma_X^2=\sigma_W^2=1$}
    \label{fig1}
\end{figure}
\vspace*{-.5cm}
\begin{figure}[!htb]
    \centering
    \includegraphics[width = \linewidth]{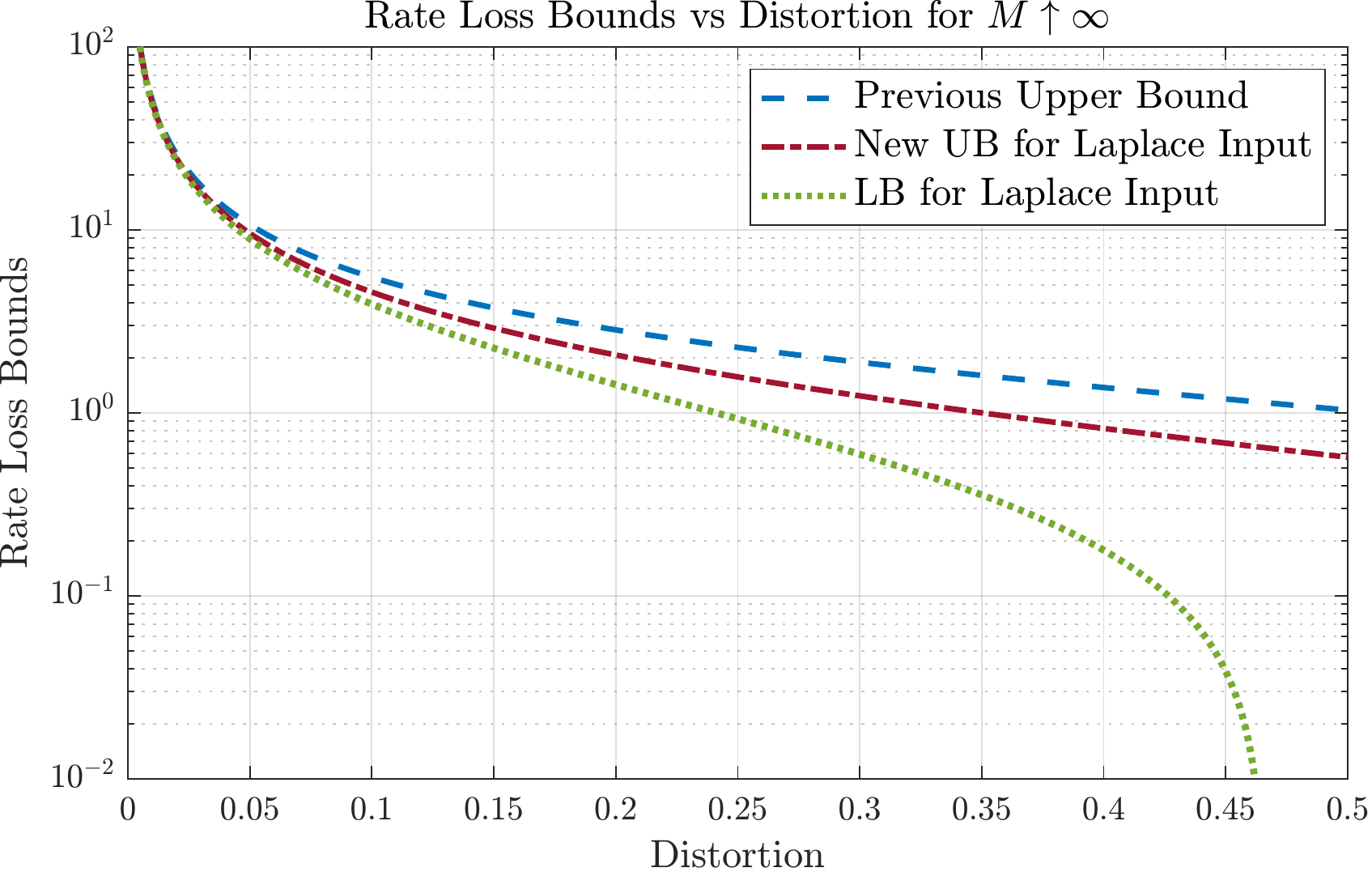}
    \caption{Comparison of Bounds for $\sigma_X^2=\sigma_W^2=1$}
    \label{fig2}
\end{figure}

\section{Extension to Exponential Families}\label{extensiontoexponentialsection}
We investigate the exponential family generalization of the differential entropy of the conditional expectation. We allow the input to be of any distribution with finite variance and differential entropy, and the output given the input is restricted to the exponential family with the canonical parameter equal to the input. 
\subsection{Natural Exponential Family Setup}

In this section, we \revision{adopt the notation used in \cite{morris,efron2011tweedie} and} replace the model of Figure~\ref{observation} by the following more general model:
\begin{equation}
X \sim q(\cdot) \; \text { and } \; Y \mid X=x \sim p_{x}(y)=e^{x y-A(x)} p_{b}(y)  \label{Eq-mod-expofam}  
\end{equation}
where $x$ is the canonical parameter of the family, $A(x)$ is the cumulant generating function (CGF) \footnote{\revision{Derivatives of CGF yield the cumulants. See \cite{morris} for further explanations.}}, $p_b(y)$ is the \revision{absolutely continuous} base measure, and \revision{$q(\cdot)$ is the PDF of the absolutely continuous random variable $X$ with known mean and variance}. For this model, the following lower bound on the differential entropy of the conditional expectation holds.
\begin{theorem}.\label{thm-expofam}
For the model of Equation~\eqref{Eq-mod-expofam}, we have
\begin{align}
    h(\E{X\given Y}) & \ge  2h(X) - h(Y) \nonumber \\
      & \,\,\,\,\,\, + 2\left(h(Y\given X)-\frac{1}{2}\loge{2\pi e}\right).\label{eq-thm-expofam}
\end{align}
\end{theorem}

\begin{proof}
The proof follows by analogy to Theorem~\ref{thm-Gauss-scalar}. 
By Bayes' Rule, the posterior density of $X$ given $Y=y$ is 
\begin{equation}
    p(x\given y) = p_x(y) q(x)/p(y) \label{bayes}
\end{equation}
where $p(y)$ is the marginal density which is calculated as 
\begin{equation}
    p(y) = \int p_x(y) q(x)\, dx.
\end{equation}
Plugging (\ref{Eq-mod-expofam}) in (\ref{bayes}), we obtain
\begin{equation}
    p(x\given y) = e^{x y-\log\frac{p(y)}{p_b(y)}} \left(q(x)\,e^{-A(x)}\right),\label{expfamconditional}
\end{equation}
which is a natural exponential family with canonical parameter $y$ and CGF $\log \frac{p(y)}{p_b(y)}$. Taking derivatives of the CGF with respect to $y$, \revision{we obtain the natural exponential family generalization of \emph{Tweedie's formula}, provided the differentiation under the integral sign is justified. We discuss these details in Appendix~\ref{generalizedTweedie}.}
\begin{align}
    \E{X\given Y=y} &= \frac{d}{dy} \loge{\frac{p(y)}{p_b(y)}} \label{efron1} \\
    \Var{X\given Y=y} &= \frac{d}{dy}\E{X\given Y=y} = \frac{d^2}{dy^2} \loge{\frac{p(y)}{p_b(y)}}. \label{efron2}
\end{align}
Hence, we have the same formula for calculating $h(\E{X\given Y})$.
Specifically, we can observe the following steps:
\begin{IEEEeqnarray}{rCl}
    h(\E{X\given Y}) &=& h(Y) + \E{\log\left\vert\frac{d}{dY} \E{X\given Y}\right\vert} \label{expfam1}\\
    &=& h(Y) + \E{\loge{\Var{X\given Y}}} \label{expfam2}\\
    &\geq& h(Y)+2h(X\given Y)-\log(2\pi e) \label{expfam3}\\
    &=& 2h(X)-h(Y) \nonumber\\&&\hspace{1.05cm}+2h(Y\given X)-\log(2\pi e) \label{expfam4}
\end{IEEEeqnarray}
where (\ref{expfam1}) follows immediately from the change of variable formula, (\ref{expfam2}) follows from (\ref{efron1}), and (\ref{efron2}). In order to justify (\ref{expfam3}), we can use the following argument. Since 
\begin{equation}
    h(X\given Y=y)\leq \frac{1}{2}\,\loge{2\pi e\,\Var{X\given Y=y}} \label{expfaminequality}
\end{equation} for every $y$, we take expectations of both sides to obtain 
\begin{equation}
    \E{\loge{\Var{X\given Y}}} \geq  2\left(h(X\given Y)-\frac{1}{2}\loge{2\pi e}\right).
\end{equation}
Finally, (\ref{expfam4}) follows from the definition. That is,
\begin{equation}
    I(X;Y) = h(X)-h(X\given Y) = h(Y)-h(Y\given X).
\end{equation}
\end{proof}
\revision{
\begin{remark}.
\begin{itemize}
    \item Note that (\ref{expfaminequality}) inherently assumes $X$ is supported on $\mathbb{R}$, and therefore, it can be tightened in case a support constraint is added to the model (\ref{Eq-mod-expofam}). We refer to \cite{maxentropyvarianceconstraint} for a detailed discussion on this matter. 
    \item From (\ref{expfaminequality}), it is evident that the equality is achieved in (\ref{eq-thm-expofam}) if and only if $X\given Y=y$ is Gaussian. Hence, the bound is tighter for posterior distributions close to Gaussian in terms of KL divergence.
\end{itemize}

\end{remark}
 We note that if we specialize the model of Equation~\eqref{Eq-mod-expofam} to the AWGN model of Figure~\ref{observation},
then the corrective term in the second line of Equation~\eqref{eq-thm-expofam} vanishes and we obtain Theorem~\ref{thm-Gauss-scalar}. In that case, $p_b(y)$ is the Gaussian density with zero mean and $\sigma_W^2$ variance and
\begin{align}
    x \triangleq \frac{\mu}{\sigma_W^2},\; A(x) \triangleq \frac{1}{2}\sigma_W^2\,x^2=\frac{1}{2}\frac{\mu^2}{\sigma_W^2}.
\end{align}
Hence, $Y\given X=x$ is a Gaussian random variable with mean $\sigma_W^2\,x$ and variance $\sigma_W^2$, and the model (\ref{Eq-mod-expofam}) corresponds to 
a modified AWGN channel $Y=\sigma_W^2 X+W$. Introducing a change of variable $\widetilde{X}\triangleq \sigma_W^2 X$ gives the desired result, i.e., $h(\E{\widetilde{X}\given Y})\geq 2h(\widetilde{X})-h(Y)$.}
In general, however, the corrective term can be positive or negative,
which is illustrated by the example given below in Section~\ref{Sec-gamma}.
\begin{remark}.
It is important to note that the presented setup is for the natural exponential family, which includes the Normal distribution with known variance, the Poisson distribution, the Gamma distribution with known shape parameter $\alpha$, the Binomial distribution with known number of trials, and the Negative Binomial distribution with known $r$. As we work with differential entropy, an important example in this setup is the Gamma distribution. 
\end{remark}

\subsection{Example : Gamma Distribution}\label{Sec-gamma}
As a concrete example of the model of Equation~\eqref{Eq-mod-expofam}, we now consider the Gamma distribution, {\it i.e.,}
we let $X$ be any positive, absolutely continuous random variable with finite variance $\sigma_X^2$ and differential entropy $h(X)$, and
\begin{equation}
X \sim p_X(\cdot)   \; \text{ and } \; Y\mid X=x \sim \frac{x^\alpha}{\Gamma(\alpha)} y^{\alpha-1} e^{-x y}.
\end{equation}
That is, the conditional distribution of $Y$ given $X=x$ is Gamma with known shape parameter $\alpha$ and rate parameter $x>0$ \footnote{\revision{For a fixed shape parameter, the canonical parameter of Gamma distribution is the additive inverse of the rate parameter. However, this does not change the Theorem~\ref{thm-expofam}}}. In this case, $h(Y\given X)$ is calculated as $-\E{\log X}+\alpha+\log\Gamma(\alpha)+(1-\alpha)\psi(\alpha)$. \revision{Define the corrective term
\begin{IEEEeqnarray}{rCl}
    \Delta(p_X,\alpha) &\triangleq& \alpha+\log\Gamma(\alpha)+(1-\alpha)\psi(\alpha)\nonumber\\&&-\frac{1}{2}\loge{2\pi e}-\E{\log X}.
\end{IEEEeqnarray}
Depending on $p_X$ and $\alpha$, $\;\Delta(p_X,\alpha)$ can be positive or negative, and Theorem~\ref{thm-expofam} evaluates to
\begin{equation}
    h(\E{X\given Y})\geq 2h(X) -h(Y)+2\Delta(p_X,\alpha).\label{gammaexamplebound}
\end{equation}} 

Observe that this model corresponds to a multiplicative model rather than the additive model we considered in Figure~\ref{observation}, i.e., 
\begin{equation}
    Y = X^{-1}\, G
\end{equation}
where $G$ is Gamma with shape parameter $\alpha$ and rate parameter 1, independent of $X$. This is illustrated in Figure~\ref{fig-gamma}.

\begin{figure}[!htb]
    \centering
    \begin{tikzpicture}
	\matrix (m1) [row sep=8mm, column sep=8mm]
	{
		\node[dspnodeopen,dsp/label=above] (m00) {$X$};    &
		\node[coordinate ]                 (m01) {};          &
		\node[dspfilter,minimum height=1.1cm,minimum width=1.8cm,text height=2.1em]                   (m02) {\footnotesize Multiplicative\\ \footnotesize Inverse}; &
		\node[coordinate ]                 (m03) {};          &
	    \node[dspmixer, dsp/label=bottom]                   (m04) {}; &
		\node[dspnodeopen,dsp/label=above] (m05) {$Y$};       \\
		&&&&\node[dspnodeopen, dsp/label=right]   (m14) {$G$};    
		\\
		\\
	};

	
	\begin{scope}[start chain]
		\chainin (m00);
		\chainin (m02) [join=by dspflow];
	\end{scope}

	\foreach \i [evaluate = \i as \j using int(\i+1),
	             evaluate = \i as \k using int(\i+2),] in {2,3}
	{
		\begin{scope}[start chain]
			\chainin (m0\i);
			\chainin (m0\j) [join=by dspline];
			\chainin (m0\k) [join=by dspline];
		\end{scope}
	}
	\draw[dspflow] (m03) -- node[midway,above] {$X^{-1}$} (m04);
	\draw[dspflow] (m14) -- node[midway,above] {} (m04);
\end{tikzpicture}
\vspace*{-1cm}
\caption{The Gamma Model}
\label{fig-gamma}
\end{figure}
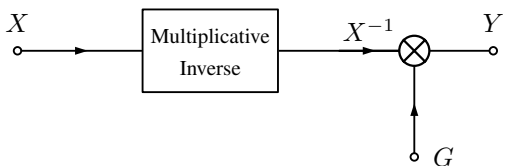
\revision{
To illustrate the Theorem~\ref{thm-expofam} for a specific input distribution, let $X$ be a Beta-prime random variable, i.e., its PDF is
\begin{equation}
    p_X(x) = \begin{cases}
    \frac{\Gamma(\alpha)}{\Gamma(\alpha-\gamma)\Gamma(\gamma)}(x-1)^{\alpha-\gamma-1}x^{-\alpha}&\text{ if } x\geq 1\\
    0&\text{ otherwise}
    \end{cases}
\end{equation}
where $\alpha>\gamma>2$ so that its variance is finite. It is easy to check that for this choice of $p_X(\cdot)$, $Y$ is Gamma with shape parameter $\gamma$ and rate parameter $1$, i.e., $
    p_Y(y) = \frac{e^{-y}\,y^{\gamma-1}}{\Gamma(\gamma)} \text{ for } y>0$. Conditioned on $Y=y$, \begin{equation} X = 1+T\end{equation} where $T$ is Gamma with shape parameter $\alpha-\gamma$ and rate parameter $y$. Thus, \begin{align}
        \E{X\given Y} &= 1+\frac{\alpha-\gamma}{Y}\\
        \Var{X\given Y} &= \frac{\alpha-\gamma}{Y^2}
    \end{align}
and every term in (\ref{eq-thm-expofam}) can be calculated analytically as a function of $\alpha$ and $\gamma$. Furthermore, the gap to the lower bound in (\ref{gammaexamplebound}) depends only on the difference $d \triangleq \alpha-\gamma$:
\begin{IEEEeqnarray}{rCl}
    h(\E{X\given Y})-\left(2h(X)-h(Y)+2\Delta\right) &=& \loge{\frac{2\pi e d}{\left(\Gamma(d)\right)^2}}\nonumber\\&& +2(d-1)\psi(d)\nonumber\\&&-2d.
\end{IEEEeqnarray}
Using the series expansion of the gamma and digamma functions, one can show that the gap is asymptotically equivalent to $\frac{2}{d}$ and $\frac{2}{3d}$ for $d\downarrow 0$ and $d\uparrow \infty$, respectively.  
We illustrate the tightness of the lower bound as a function of $d$ in Figure~\ref{fig:d}. 
\begin{figure}[!htb]
    \centering
    \includegraphics[width = \linewidth]{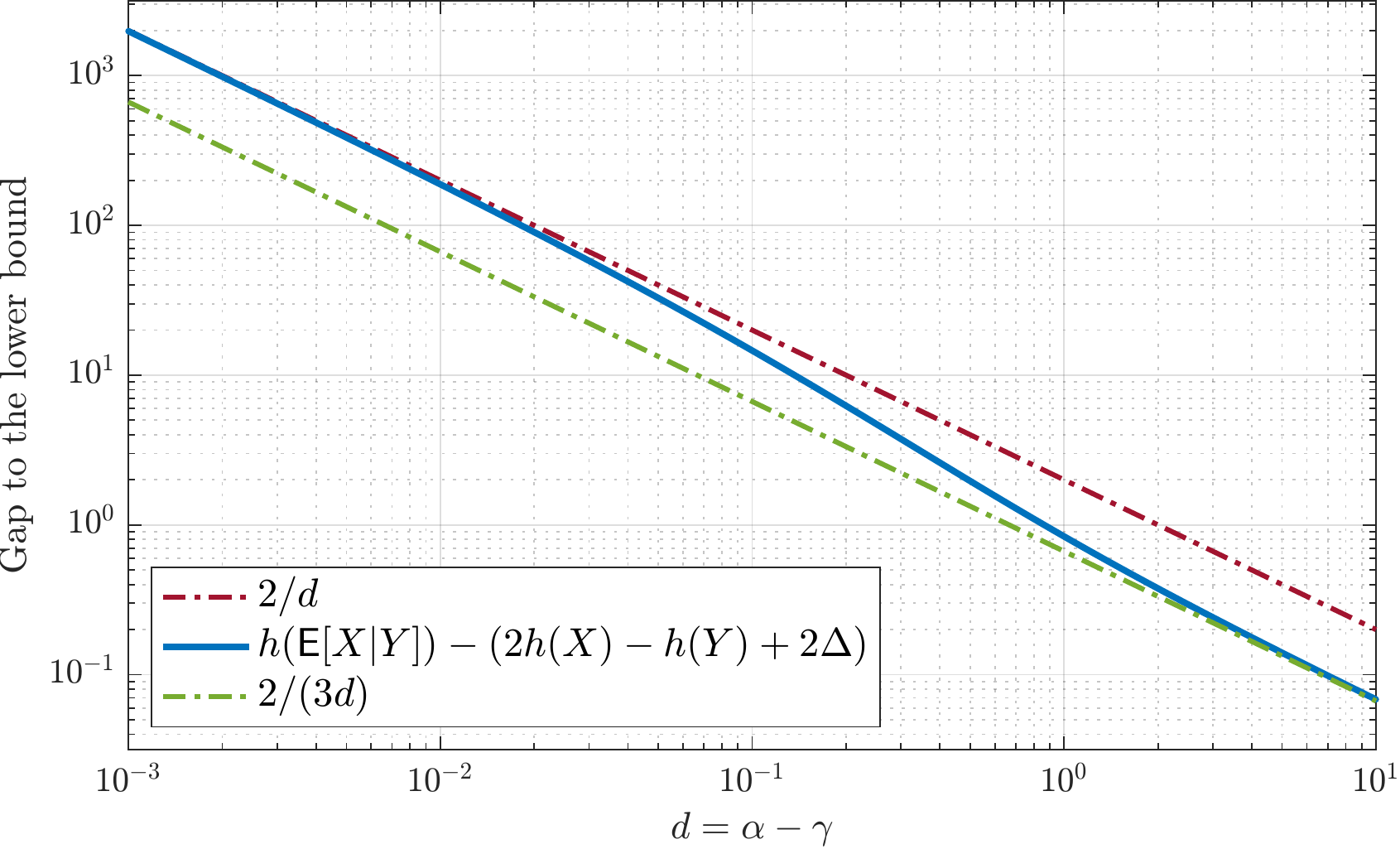}
    \caption{Gap to the lower bound vs $d$}
    \label{fig:d}
\end{figure}
}
\section{Extension to the Vector Case}\label{extensiontovectorsection}
In this section, we consider the extension of the main result (\ref{mainresult}) and the upper bound (\ref{ub1}) under the vector Gaussian noise model, i.e., the input-output relationship is governed by 
\begin{equation}
\mathbf{Y} = \mathbf{A}\mathbf{X} + \mathbf{W}
\label{vectorGauss}
\end{equation}
where $\mathbf{W}\in \mathbb{R}^n$ is a zero mean Gaussian random vector with positive-definite covariance matrix $\mathbf{K_W}$, $\mathbf{X}\in \mathbb{R}^n $ is and $\mathbf{A}$ is a full-rank, $n\times n$ matrix. It is assumed that $\mathbf{X}$ and $\mathbf{W}$ are independent, and the only assumption on $\mathbf{X}$ is that its covariance matrix $\mathbf{K_X}$ is full-rank, i.e., $\mathbf{X}$ is non-degenerate. 
\par Denote the conditional variance matrix by 
\begin{align}
    \Varvec{\mathbf{X}\given \mathbf{Y}} &\triangleq \E{\mathbf{X}\mathbf{X}^T\given \mathbf{Y}}-\E{\mathbf{X}\given \mathbf{Y}}\E{\mathbf{X}^T\given \mathbf{Y}},
\end{align} and the MMSE matrix by 
\begin{equation}
    \boldsymbol{\mathsf{MMSE}}(\mathbf{X}\given\mathbf{Y})\triangleq \E{\Varvec{\mathbf{X}\given \mathbf{Y}}},
\end{equation} and the Jacobian matrix of a transformation $\boldsymbol\phi:\mathbf{R}^n\mapsto \mathbf{R}^m$ by $\mathbf{J}_\mathbf{y}\boldsymbol{\phi}(\mathbf{y})$ with $i^{th}$ row, $j^{th}$ column element being $\frac{\partial \phi_j}{\partial y_i}$. The result of this section is that a similar lower bound is valid under the model in (\ref{vectorGauss}).
\begin{proposition}\hspace*{.25em}(A Lower Bound of Differential Entropy in Vector Case):
For the model given in (\ref{vectorGauss}), we have 
\begin{equation}
h\left(\E{\mathbf{X}\given \mathbf{Y}}\right) \geq 2h\left(\mathbf{X}\right) -h\left(\mathbf{Y}\right) +\log\det \mathbf{A}.
\label{mainresultvec}
\end{equation}
\revision{Furthermore, equality is achieved when $X$ is Gaussian.}
\end{proposition}
\begin{proof}
Since $I(\mathbf{X};\mathbf{Y}) = h(\mathbf{Y})-h(\mathbf{W}) = h(\mathbf{X})-h(\mathbf{X\given \mathbf{Y}})$, (\ref{mainresultvec}) is equivalent to 
\begin{IEEEeqnarray}{rCl}
h(\mathbf{X}\given \mathbf{Y}) &\leq &h(\mathbf{W}) -\frac{1}{2}h(\mathbf{Y})\nonumber\\&&{}+
\frac{1}{2}h(\E{\mathbf{X}\given \mathbf{Y}}) - \frac{1}{2}\log\det\mathbf{A}\\
&= &\frac{1}{2}\loge{(2\pi e)^n\,\det\mathbf{K_W}}\nonumber-\frac{1}{2}h(\mathbf{Y})\\&&+
\frac{1}{2}h(\E{\mathbf{X}\given \mathbf{Y}})-\frac{1}{2}\log\det\mathbf{A}.
\end{IEEEeqnarray}
By the maximum entropy argument, 
\begin{equation}
    h(\mathbf{X}\given \mathbf{Y}) \leq \frac{1}{2}\E{\loge{(2\pi e)^n\,\det\Varvec{\mathbf{X}\given \mathbf{Y}}}}.\label{maxentropyvector}
\end{equation}
Hence, it is sufficient to show that 
\begin{multline}
    \frac{1}{2}\E{\loge{(2\pi e)^n\,\det\Varvec{\mathbf{X}\given \mathbf{Y}}}} \\\leq \frac{1}{2}\loge{(2\pi e)^n\,\det(\mathbf{A}^{-1}\,\mathbf{K_W})}-\frac{1}{2}h(\mathbf{Y})\\+
\frac{1}{2}h(\E{\mathbf{X}\given \mathbf{Y}})
\end{multline}
which is equivalent to 
\begin{multline}
     \E{\log{\det\left(\mathbf{A}\,\mathbf{K}^{-1}_\mathbf{W}\,\Varvec{\mathbf{X}\given \mathbf{Y}} \right)}} \leq h(\E{\mathbf{X}\given \mathbf{Y}})-h(\mathbf{Y}).
\end{multline}
By the change of variables formula, 
\begin{equation}
    h(\E{\mathbf{X}\given \mathbf{Y}})  = h(\mathbf{Y})+ \E{\log{\left\vert\det\mathbf{J}_\mathbf{Y}(\E{\mathbf{X}\given \mathbf{Y}})\right\vert}}
\end{equation} provided that the transformation $\mathbf{y}\mapsto \E{\mathbf{X}\given \mathbf{Y} = \mathbf{y}}$ is diffeomorphic. Thus, it remains to justify 
\begin{equation}
    \E{\log{\left\vert\det\mathbf{J}_\mathbf{Y}(\E{\mathbf{X}\given \mathbf{Y}})\right\vert}}\ge \E{\log{\det\left(\mathbf{A}\,\mathbf{K}^{-1}_\mathbf{W}\,\Varvec{\mathbf{X}\given \mathbf{Y}} \right)}} \label{vecproof1}
\end{equation} and $\mathbf{y}\mapsto \E{\mathbf{X}\given \mathbf{Y} = \mathbf{y}}$ is diffeomorphic. Under Gaussian noise, the variance identity of Hatsell and Nolte \cite{dytso_poor_shitz_2020} gives 
\begin{align}
    \mathbf{J}_\mathbf{y}\E{\mathbf{X}\given \mathbf{Y}=\mathbf{y}}=\mathbf{A}^{-1}\,\mathbf{K}^{-1}_\mathbf{W}\,\mathbf{A}\,\Varvec{\mathbf{X}\given \mathbf{Y}=\mathbf{y}}\mathbf{A}^T \label{vecproof2}
\end{align}
for every realization $\mathbf{y}$. Hence, (\ref{vecproof1}) is satisfied with equality and $\mathbf{y}\mapsto \E{\mathbf{X}\given \mathbf{Y} = \mathbf{y}}$ is diffeomorphic provided that $\mathbf{K_X}$ is full-rank.
\revision{When the input is Gaussian, (\ref{maxentropyvector}) and therefore the main lower bound (\ref{mainresultvec}) are satisfied with equality: 
\begin{align*}
h(\E{\mathbf{X}\given\mathbf{Y}})
    &= 2h\left(\mathbf{X}\right) -h\left(\mathbf{Y}\right) +\log\det \mathbf{A} \\&= \tfrac{1}{2}\log\det(2\pi e (\mathbf{A}\mathbf{K_X})^{\scalebox{0.7}{2}}(\mathbf{A}\mathbf{K_X}\mathbf{A}^T+\mathbf{K_W})^{\scalebox{0.7}{-1}}).
\end{align*}}
\end{proof}
By the concavity of the log determinant, Jensen's Inequality gives an upper bound to $h(\E{\mathbf{X}\given \mathbf{Y}})$.
\begin{lemma}{\hspace*{.25em}(Upper Bounds of Differential Entropy in Vector Case):}
For the model given in (\ref{vectorGauss}), we have 
\begin{IEEEeqnarray}{rCl}
    h(\E{\mathbf{X}\given \mathbf{Y}}) &= &h(\mathbf{Y}) \nonumber\\&&+ \E{\log\det \left(\mathbf{A}\,\mathbf{K}^{-1}_\mathbf{W}\,\Varvec{\mathbf{X}\given \mathbf{Y}} \right)} \label{ubv1}\\
                  &\leq &h(\mathbf{Y}) + \log\det\E{\Varvec{\mathbf{X}\given \mathbf{Y}}}\nonumber\\&&+\log\det \left(\mathbf{A}\,\mathbf{K}_\mathbf{W}^{-1}\right) \label{ubv2}\\
                  &= &h(\mathbf{Y}) +\log\det \boldsymbol{\mathsf{MMSE}}(\mathbf{X}\given\mathbf{Y})\nonumber\\&&+\log\det \left(\mathbf{A}\,\mathbf{K}_\mathbf{W}^{-1}\right) \label{ubv3}\\
                  &\leq &\frac{1}{2}\loge{(2\pi e)^n\,\det (\mathbf{K}_\mathbf{X}+\mathbf{K}_\mathbf{W})}\nonumber\\&&
                  +\log\det \boldsymbol{\mathsf{MMSE}}(\mathbf{X}\given\mathbf{Y})\nonumber\\&&+\log\det \left(\mathbf{A}\,\mathbf{K}_\mathbf{W}^{-1}\right) \label{ubv4}
\end{IEEEeqnarray}
\end{lemma}
\begin{proof} (\ref{ubv1}) follows from Hatsell and Nolte Identity \cite{dytso_poor_shitz_2020} combined with the change of variables,  (\ref{ubv2}) follows from Jensen's Inequality, (\ref{ubv3}) is by the definition of MMSE matrix and (\ref{ubv4}) is by the maximum entropy argument. \end{proof}

\begin{remark}{\hspace{.25em}(A specific case):}
One could set $\mathbf{A} = \mathbf{I}_n$, i.e., the simple extension of the scalar model in the vector setting:  
\begin{equation}
    \mathbf{Y} = \mathbf{X}+ \mathbf{W}
\end{equation}
for a Gaussian random vector $\mathbf{W}$ with covariance matrix $\mathbf{K}_\mathbf{W}$, and an arbitrary random vector $\mathbf{X}$ with finite covariance matrix $\mathbf{K}_\mathbf{X}$.
In this setup, the main inequality becomes 
\begin{equation}
    h\left(\E{\mathbf{X}\given \mathbf{Y}}\right) \geq 2h\left(\mathbf{X}\right) -h\left(\mathbf{Y}\right), \label{vecsimple1}
\end{equation}
which is the same bound as in the scalar case, cf. (\ref{mainresult}). Similarly, the upper bound (\ref{ubv3}) becomes 
\begin{multline}
    h\left(\E{\mathbf{X}\given \mathbf{Y}}\right) \leq h(\mathbf{Y}) +\log\det \boldsymbol{\mathsf{MMSE}}(\mathbf{X}\given\mathbf{Y})\\+\log\det \left(\mathbf{K}_\mathbf{W}^{-1}\right),
\end{multline}
which is an extension of (\ref{ub2}).
\end{remark}
\revision{
\begin{remark}.
Similar to the scalar case, (\ref{maxentropyvector}) and therefore the main lower bound (\ref{mainresultvec}) are tight for input distributions close to Gaussian in terms of KL divergence. For a fixed input distribution, if the input or the additive noise is scaled so that $\mathbf{K_X}$ and $\mathbf{K_W}$ vary, (\ref{mainresultvec}) becomes tighter as $\det \left(\mathbf{K_X}\, \mathbf{K}_{\mathbf{W}}^{-1}\right)$ increases.
\end{remark}}

\begin{appendices}
\section{Derivation of Hatsell-Nolte Identity}\label{app-HN}
As the original proofs are not explicit, we include derivations for Tweedie's Formula and for the Hatsell-Nolte identity relying on multiple uses of differentiation under the integral sign. Hence, the tool we need for this proof is the well-known \emph{Leibniz Rule} which we state here for completeness and future reference:

\begin{theorem}\hspace*{.25em}(Leibniz Integral Rule):\cite{durrett}\label{thm-Leibniz}
Let $(S, \mathcal{S}, \mu)$ be a measure space. Let $f$ be a complex valued function defined on $\mathbbm{R} \times S .$ Let $\delta>0$, and suppose that for $y \in(x-\delta, x+\delta)$ we have
\begin{enumerate}[label=(\roman*)]
    \item $p_Y(y)=\int_{S}\, \xi(y, s)\, \mu(d s) \text { with } \int_{S}|\xi(y, s)|\, \mu(d s)<\infty$.
    \item For fixed s, $\frac{\partial \xi}  {\partial y}\,(y, s)$ exists and is a continuous function of $y$.
    \item $\int_{S}\, \sup _{\theta \in[-\delta, \delta]}\left|\frac{\partial \xi}{\partial y}(x+\theta, s)\right| \mu(d s)<\infty$.
\end{enumerate}
Then,  $p_Y^{\prime}(x)=\int_{S}\,\frac{\partial \xi}{\partial x} (x,s)\,\mu(ds)$.
\end{theorem}

As we use Tweedie's Formula in the proof of Hatsell-Nolte Identity, we begin by proving the former.
\begin{lemma}{\hspace*{.25em}(Tweedie's Formula):}
For the model given in (\ref{themodel}), 
\begin{equation}
    \label{tweedie}
   \E{X\,\vert\, Y=y} = y +  \sigma_W^2\frac{d}{dy}\,\log p_Y(y).
\end{equation}
\end{lemma}
\begin{proof}
Let $\phi_W$ denote the zero-mean Gaussian PDF with variance $\sigma_W^2.$
By independence, the density $p_Y$ is the convolution of $p_X$ and $\phi_W$:
\begin{equation}
\label{conv}
    p_Y(y) = \int_{\mathbbm{R}}p_X(s)\,\phi_W(y-s)\,ds.
\end{equation}
Multiplying both sides by $\sigma_W^2$ and taking derivative of (\ref{conv}) w.r.t. $y$, we obtain 
\begin{align}
\label{firstder}
        \sigma_W^2 p_Y^{\prime}(y) &=\sigma_W^2 \frac{d}{dy}\int_{\mathbbm{R}}p_X(s)\,\phi_W(y-s)\,ds\\
        &\labelrel={firstder:1} \int_{\mathbbm{R}}p_X(s)\,(s-y)\,\phi_W(y-s)\,ds\\
        &= \label{firstderfinal} \int_{\mathbbm{R}}s\,p_X(s)\,\phi_W(y-s)\,ds - y\,p_Y(y)
\end{align}
where we used the fact that $\sigma_W^2\frac{d}{ds}\,\phi_W(s) = -s\,\phi_W(s),$ and Step \eqref{firstder:1} follows from the Leibniz Integral Rule, as we argue carefully below. Let us divide both sides of (\ref{firstderfinal}) by $p_Y(y)$:
\begin{align}
\label{tweediefinal}
   \sigma_W^2 \frac{p_Y^{\prime}(y)}{p_Y(y)} &=\sigma_W^2\frac{d}{dy}\,\log p_Y(y) \\&=  \int_{\mathbbm{R}}s\,\frac{p_X(s)\,\phi_W(y-s)}{p_Y(y)}\,ds - y.
\end{align}
Observe that $\E{X\,\vert\,Y = y} = \int_{\mathbbm{R}}s\,\frac{p_X(s)\,\phi_W(y-s)}{p_Y(y)}\,ds$ since the joint density of $X,Y$ is simply $p_{X,Y}(x,y) = p_X(x)\,\phi_W(y-x)$. Hence, we obtain the desired result \begin{equation}\E{X\,\vert\, Y=y} = y + \sigma_W^2\frac{d}{dy}\,\log p_Y(y).\label{tweediederived}\end{equation}
\par For the justification of Step \eqref{firstder:1}, one can use the Leibniz Integral Rule stated above in Theorem~\ref{thm-Leibniz}.
To use this theorem, we now verify that Conditions {\it (i), (ii),} and {\it (iii)} are satisfied.
First, observe that
$p_Y(y) = \int_{\mathbbm{R}}\phi_W(y-s)\,p_X(s)\,ds \implies \xi(y,s) = \phi_W(y-s) \text{ and } \mu(ds) = p_X(s)\,ds.$
With this, we observe the following:
\begin{enumerate}[label=(\it\roman*)]
    \item Since $\xi(y,s) = |\xi(y,s)| \text{ for all } y,s \in \mathbbm{R}$; we have $\int_\mathbbm{R} |\xi(y,s)|\,\mu(ds) = p_Y(y)<\infty$.
    \item For any fixed $s$, $\frac{\partial \xi}  {\partial y}\,(y, s) = \sigma_W^{-2} (s-y)\,\phi_W(y-s)$ obviously exists and is a continuous function of $y$ for every $y\in\mathbbm{R}$.
    \item 
    \begin{align}
        \left|\frac{\partial \xi}  {\partial y}\,(x+\theta, s)\right| &= \sigma_W^{-2}\left |(s-x-\theta)\,\phi_W(x+\theta-s)\right| \\&\leq \dfrac{(|s|+|x|+|\theta|)}{ \sigma_W^{3}\sqrt{2\,\pi}}
    \end{align}
     for all $s,\,x,\,\theta \in\mathbbm{R}$. Hence, 
    \begin{align}
        \lefteqn{\int_{S}\; \sup _{\theta \in[-\delta, \delta]}\left|\frac{\partial \xi}{\partial y}(x+\theta, s)\right| \mu(d s)}\\
        &\leq \int_{S} \dfrac{|s|+|x|+\delta}{\sigma_W^3\sqrt{ 2\,\pi}}\,\mu(d s)\\&<\infty
    \end{align}
    
    for every $\delta>0$ and $x\in\mathbbm{R}$ since we assume that $X$ is square-integrable, thus integrable. Observe that (iii) is satisfied for every $\delta>0$, hence we have 
    \begin{align}
        p_Y^{\prime}(y) &= \dfrac{d}{dy}\int_{\mathbbm{R}}\phi_W(y-s)\,p_X(s)\,ds \\&= \int_{\mathbbm{R}}\dfrac{\partial}{\partial y}\phi_W(y-s)\,p_X(s)\,ds \quad \forall y\in\mathbbm{R}.
    \end{align}
\end{enumerate}
This concludes the proof that the Leibniz Integral rule applies, and thus, concludes the proof of Tweedie's formula.
\end{proof}

\begin{theorem}{\hspace*{.25em}(Hatsell-Nolte Identity):}\label{thm-hatsell-nolte}
For the model given in (\ref{themodel}),
\label{HatsellNolteDer}
\begin{equation}
\label{hn}
\sigma_W^2\frac{d}{dy}\E{X\,\vert\,Y=y} = \Var{X\,\vert\,Y=y}.
\end{equation}
\end{theorem}
\begin{proof}
By definition, $\text{Var}\left(X\,\vert\,Y=y\right) = \E{X^2\,\vert\,Y=y} - \E{X\,\vert\,Y=y}^2$. Hence, it is sufficient to derive a formula for $\E{X^2\,\vert\,Y=y}$ by multiplying (\ref{firstderfinal}) by $\sigma_W^2$ and taking derivative w.r.t $y$:

\begin{IEEEeqnarray}{rCl}
\label{secondder}
    \sigma_W^4 p_Y^{\prime}(y) &= &\sigma_W^2\int_{\mathbbm{R}}s\,p_X(s)\,\phi_W(y-s)\,ds - \sigma_W^2 y\,p_Y(y) \nonumber\\
    \sigma_W^4 p_Y^{\prime\prime}(y) &\labelrel={secondder:2} \nonumber  &\int_{\mathbbm{R}}s\,(s-y)\,p_X(s)\,\phi_W(y-s)\,ds\\&&{}-\sigma_W^2p_Y(y)-\sigma_W^2 y\,p_Y^{\prime}(y)\\
    &= &\int_{\mathbbm{R}}s^2\,p_X(s)\,\phi_W(y-s)\,ds\nonumber\\&&{}-y\int_{\mathbbm{R}}s\,p_X(s)\,\phi_W(y-s)\,ds\nonumber\\&&{}-\sigma_W^2 p_Y(y)-\sigma_W^2 y\,p_Y^{\prime}(y). \label{secondderfinal}
\end{IEEEeqnarray}

Dividing both sides of (\ref{secondderfinal}) by $p_Y(y)$ and recalling $\E{X^p\,\vert\,Y = y} = \int_{\mathbbm{R}}s^p\,\frac{p_X(s)\,\phi_W(y-s)}{p_Y(y)}\,ds\;$, we obtain
\begin{IEEEeqnarray}{rCl}
\label{almostthere}
    \sigma_W^4\frac{p_Y^{\prime\prime}(y)}{p_Y(y)} &=& \E{X^2\,\vert\,Y = y} - y\, \E{X\,\vert\,Y = y} \nonumber\\&&{}-\sigma_W^2\left(1+y\,\frac{p_Y^{\prime}(y)}{p_Y(y)}\right)\\
    &=& \E{X^2\,\vert\,Y = y} - y\left(y+\sigma_W^2\frac{p_Y^{\prime}(y)}{p_Y(y)}\right) \nonumber \\&&{}-\sigma_W^2\left(1+y\,\frac{p_Y^{\prime}(y)}{p_Y(y)}\right),
\end{IEEEeqnarray}
which implies 
\begin{equation}
\label{exsq}
    \E{X^2\,\vert\,Y = y} = \sigma_W^4\frac{p_Y^{\prime\prime}(y)}{p_Y(y)}+2\sigma_W^2y\,\frac{p_Y^{\prime}(y)}{p_Y(y)}+y^2+\sigma_W^2.
\end{equation}
Combining (\ref{tweediederived}) and (\ref{exsq}), we obtain
\begin{IEEEeqnarray}{rCl}
    \text{Var}\left(X\,\vert\,Y=y\right)&=&
    \sigma_W^4\frac{p_Y^{\prime\prime}(y)}{p_Y(y)}+2\sigma_W^2y\,\frac{p_Y^{\prime}(y)}{p_Y(y)}+y^2+\sigma_W^2\nonumber\\&&{}- \left(y+\sigma_W^2\frac{p_Y^{\prime}(y)}{p_Y(y)}\right)^2\\
    &=& \sigma_W^4\frac{p_Y^{\prime\prime}(y)\,p_Y(y)-p_Y^{\prime}(y)^2}{p_Y(y)^2}+\sigma_W^2\\
    &=& \sigma_W^2\frac{d}{dy} \left(y+\sigma_W^2\frac{p_Y^{\prime}(y)}{p_Y(y)}\right)\\
    &=& \sigma_W^2\frac{d}{dy} \E{X\,\vert\,Y=y}.
\end{IEEEeqnarray}

\par What remains is to justify the step \eqref{secondder:2}, which follows from the Leibniz Integral Rule
stated above in Theorem~\ref{thm-Leibniz}.
To use this theorem, we verify that Conditions {\it (i), (ii),} and {\it (iii)} are satisfied.
In this setting, 
\begin{equation*}
    p_Y^{\prime}(y) = \sigma_W^{-2}\int_{\mathbbm{R}}(s-y)\,\phi_W(y-s)\,p_X(s)\,ds
\end{equation*}$ \implies \xi(y,s) = \sigma_W^{-2}(s-y)\, \phi_W(y-s) \text{ and } \mu(ds) = p_X(s)\,ds.$ With this, we observe the following:
\begin{enumerate}[label=(\it\roman*)]
    \item $\int_\mathbbm{R} |\xi(y,s)|\,\mu(ds) = \int_\mathbbm{R} \sigma_W^{-2}|s-y|\,\phi_W(y-s)\,\mu(ds)\leq \int_\mathbbm{R} (|s|+|y|)\,\frac{1}{\sigma_W^3\sqrt{2\,\pi}}\,\mu(ds)<\infty$ for all $y\in\mathbbm{R}$.
    \item For any fixed $s$, 
    \begin{equation*}
        \frac{\partial \xi}  {\partial y}\,(y, s) =\sigma_W^{-4}
    \left((s-y)^2-\sigma_W^2\right)\,\phi_W(y-s)
    \end{equation*} obviously exists and is a continuous function of $y$ for every $y\in\mathbbm{R}$.
    \item 
    \begin{IEEEeqnarray}{rCl}
    \hspace*{-1cm}
        \left|\frac{\partial \xi}  {\partial y}\,(x+\theta, s)\right| &= &\sigma_W^{-4}\left|(s-x-\theta)^2-\sigma_W^2\right|\,\left|\phi_W(x+\theta-s)\right| \nonumber\\&\leq &\frac{\left|s^2-2\,(x+\theta)\,s+(x+\theta)^2-\sigma_W^2\right|}{\sigma_W^5\sqrt{2\,\pi}}\nonumber\\&\leq &\frac{s^2+2\,(|x|+|\theta|)\,|s|}{\sigma_W^5\sqrt{2\,\pi}}\nonumber\\&&\hspace{1.25
        cm}+\frac{x^2+\theta^2+2\,|x|\,|\theta|+\sigma_W^2}{\sigma_W^5\sqrt{2\,\pi}}\nonumber
    \end{IEEEeqnarray}
    for all $s,\,x,\,\theta \in \mathbbm{R}$.
    Hence, 
    \begin{IEEEeqnarray}{rCl}
        &\int_{S}\; &\sup_{\theta \in[-\delta, \delta]}\left|\frac{\partial \xi}{\partial y}(x+\theta, s)\right| \mu(d s)\nonumber\\&&{}\leq \nonumber\int_{S} \Big(\dfrac{s^2+2\,(|x|+\delta)\,|s|}{\sigma_W^5\sqrt{2\,\pi}}\\&&\nonumber\hspace{1.12cm}+\dfrac{x^2+\delta^2+2\,|x|\,\delta+\sigma_W^2}{\sigma_W^5\sqrt{2\,\pi}}\Big)\,\mu(d s)\\&&{}<\infty \nonumber
    \end{IEEEeqnarray}
    for every $\delta>0$ and $x\in\mathbbm{R}$ since we assume that $X$ is square-integrable. Since (iii) is satisfied for all $\delta>0$, 
    \begin{align*}
        p_Y^{\prime\prime}(y) &=\dfrac{d}{dy}\,p_Y^{\prime}(y) \\&= \int_{\mathbbm{R}}\,\dfrac{\partial}{\partial y}\left(\sigma_W^{-2}(s-y)\,\phi_W(y-s)\right)\,p_X(s)\,ds.
    \end{align*}
\end{enumerate}
This concludes the proof that the Leibniz Integral rule applies, and thus, concludes the proof of the Hatsell-Nolte Identity.
\end{proof}

\section{Proof of Lemma~\ref{lemma-main}}\label{proofofidentity}
To establish Lemma~\ref{lemma-main}, the starting point is the well-known formula for the differential entropy of a transformed random variable, which we state here for completeness.

\revision{
\begin{lemma}{\hspace*{.25em}(Differential Entropy of Diffeomorphic Transformations of a Random Variable):}
Let $y \mapsto \varphi(y)$ be a $C^1-$diffeomorphic\footnote{$\varphi(\cdot)$ and $\varphi^{-1}(\cdot)$ are continuously differentiable.} transformation on $\mathbb{R}$, and $Y$ be an absolutely continuous random variable with PDF $p_Y(\cdot)$ and finite differential entropy $h(Y)$. Then, the differential entropy of $\varphi(Y)$ satisfies the following equation.
\begin{equation}
    h(\varphi(Y)) = h(Y) + \E{\log \abs{\frac{d\varphi(Y)}{dY}}}.
    \label{transformation_general}
\end{equation}
\end{lemma}
\begin{proof}
Since $\varphi(\cdot)$ is $C^1-$diffeomorphic, the probability density function of $\varphi(Y)$, denoted as $p_{\varphi(Y)}(\cdot)$, can be written as 
\begin{equation}
    p_{\varphi(Y)}(\phi) = \frac{p_{Y}(\varphi^{-1}(\phi))}{|\varphi^{\prime}(\varphi^{-1}(\phi))|}. \label{transformpdf}
\end{equation}
Using (\ref{transformpdf}) and a change of variable $y\triangleq \varphi^{-1}(\phi)$, $h(\varphi(Y))$ is simplified to 
\begin{align}
     h(\varphi(Y)) &= -\int_{-\infty}^{\infty} p_{\varphi(Y)}(\phi)\log{p_{\varphi(Y)}(\phi)}\,d\phi\nonumber\\
     &= -\int_{-\infty}^{\infty} \frac{p_{Y}(\varphi^{-1}(\phi))}{|\varphi^{\prime}(\varphi^{-1}(\phi))|}\log{\frac{p_{Y}(\varphi^{-1}(\phi))}{|\varphi^{\prime}(\varphi^{-1}(\phi))|}}\,d\phi\nonumber\\
     &= -\int_{-\infty}^{\infty} \frac{p_{Y}(y)}{|\varphi^{\prime}(y)|}\log{\frac{p_{Y}(y)}{|\varphi^{\prime}(y)|}}\,\varphi^{\prime}(y)dy\nonumber\\
     &= -\int_{-\infty}^{\infty} p_{Y}(y)\log{\frac{p_{Y}(y)}{|\varphi^{\prime}(y)|}}\,dy\label{transformpdffinal}\\
     &= h(Y)+\E{\log \abs{\frac{d\varphi(Y)}{dY}}}
\end{align}
where (\ref{transformpdffinal}) follows from the $C^1-$diffeomorphic assumption. That is, $\varphi(\cdot)$ must be a strictly monotonic function, and analyzing two cases \footnote{increasing and decreasing $\varphi(\cdot)$} separately both yield (\ref{transformpdffinal}). 
\end{proof}
We apply this lemma to the function $y\mapsto \E{X\given Y=y}$ in the additive white Gaussian noise model.
It is straightforward to show that in this case, the function $y \mapsto \varphi(y)$ and its inverse are real-analytic provided that $X$ is a non-degenerate random variable (see for example Lemma 2 and Lemma 3 of \cite{dytso_poor_shitz_2020}). Therefore, $y \mapsto \varphi(y)$ is $C^\infty-$diffeomorphic and (\ref{transformation_general}) can be applied whenever $\sigma_X^2 > 0.$ Using (\ref{transformation_general}), we immediately obtain 
\begin{align}
    h(\E{X\given Y}) &= h(Y) + \E{\log{\left|\frac{d\E{X\given Y}}{dY}\right|}}\\
    &= h(Y) + \E{\loge{\frac{1}{\sigma_W^2}\Var{X\given Y}}}
\end{align}
where the last line follows from the Hatsell-Nolte identity.}

\revision{\section{Exponential Family Generalization of Tweedie's Formula}\label{generalizedTweedie}
Our objective in this section is to state and prove the Exponential Family generalization of \emph{Tweedie's Formula}, which is also based on multiple uses of differentiation under the integral sign. 
\begin{lemma}{\hspace*{.25em}(Tweedie's Formula for Exponential Family):}\\
For the model in (\ref{Eq-mod-expofam}), define $\nu(y) \triangleq \frac{p(y)}{p_b(y)}$. Suppose the following conditions hold for every $y$ in the support of $\nu(y)$. 
\begin{subequations}
\begin{align}
    \nu(y)&<\infty, \label{tweedie_exp_fam_cond:1}\\
    \E{\left|X\right|\given Y=y}&<\infty,\label{tweedie_exp_fam_cond:2}\\
    \E{X^2\given Y=y}&<\infty.\label{tweedie_exp_fam_cond:3} 
\end{align}
\end{subequations}
Then, we have 
\begin{align}
    \E{X\given Y=y} &= \frac{d}{dy}\log\nu(y),\\
    \Var{X\given Y=y} &= \frac{d^2}{dy^2}\log\nu(y).
\end{align}
\end{lemma}
\begin{proof}
 By (\ref{expfamconditional}), we have 
\begin{align}
    \E{X\given Y=y} &= \int x\,p(x\given y)\,dx\nonumber\\
    &= \int x\, e^{x y-\log\nu(y)} \left(q(x)\,e^{-A(x)}\right)\,dx.\label{condexpexpfam}
\end{align}
Taking the logarithmic derivative of $\nu(y)$, we obtain 
\begin{align}
    \frac{d\log\nu(y)}{dy} &=  \frac{\frac{d\nu(y)}{dy}}{\nu(y)}\nonumber\\&=\frac{\frac{d}{dy}\int e^{x y-A(x)}q(x)\,dx}{\nu(y)}\nonumber\\&\labelrel={justifyExpFam1}\frac{\int \frac{d}{dy}\, \left(e^{x y-A(x)}q(x)\right)\,dx}{\nu(y)}\nonumber\\
    &= \int x\,e^{xy-\log\nu(y)}\left(q(x)\,e^{-A(x)}\right)\,dx,\label{expfamcondexp:2}
\end{align}
which is the same expression as (\ref{condexpexpfam}). To justify \eqref{justifyExpFam1}, we can use Theorem~\ref{thm-Leibniz}. By (\ref{tweedie_exp_fam_cond:1}), 
\begin{align}
    \int e^{x y-A(x)}q(x)\,dx &= \int \nu(y)\,p(x\given y)\,dx\\&=\nu(y)\int p(x\given y)\,dx\\&<\infty,
\end{align}
so $(i)$ is satisfied. Checking the condition $(ii)$ is trivial, i.e., for every $x$ in the support of $q(\cdot)$, $\,y\mapsto e^{xy}$ is continuous for every $y$ in the support of $\nu(y)$. Finally, the condition $(iii)$ holds due to (\ref{tweedie_exp_fam_cond:1}) and (\ref{tweedie_exp_fam_cond:2}). Differentiating  (\ref{expfamcondexp:2}) once more, we obtain 
\begin{IEEEeqnarray}{rCl}
\frac{d^2\log\nu(y)}{dy^2} &=&\frac{d}{dy}\int x\,e^{xy-\log\nu(y)}\left(q(x)\,e^{-A(x)}\right)dx\nonumber\\&\labelrel={justifyExpFam2}&\int \frac{d}{dy} \left(x\,e^{xy-\log\nu(y)}\left(q(x)\,e^{-A(x)}\right)\right)dx\nonumber\\
    &=& \int x^2\,e^{xy-\log\nu(y)}\left(q(x)\,e^{-A(x)}\right)dx\nonumber\\
    &&-{}\int x\frac{\nu^{\,\prime}(y)}{\nu(y)}\,e^{xy-\log\nu(y)}\left(q(x)\,e^{-A(x)}\right)dx\nonumber\\
    &=& \E{X^2\given Y=y}-\left(\E{X\given Y=y}\right)^2.
\end{IEEEeqnarray}
Justification of \eqref{justifyExpFam2} is similar to \eqref{justifyExpFam1}: $(i)$ holds due to (\ref{tweedie_exp_fam_cond:1}) and (\ref{tweedie_exp_fam_cond:2}), $(ii)$ follows by
the same reasoning as \eqref{justifyExpFam1}, and $(iii)$ holds due to (\ref{tweedie_exp_fam_cond:1}) and (\ref{tweedie_exp_fam_cond:3}).
\end{proof}
\begin{remark}.
The Gamma distribution example given in Sec.~\ref{Sec-gamma} satisfies (\ref{tweedie_exp_fam_cond:1}), (\ref{tweedie_exp_fam_cond:2}), (\ref{tweedie_exp_fam_cond:3}). That is,

\begin{align}
    \nu(y) &= \int_{0}^{\infty}x^\alpha e^{-xy}\, q(x)\,dx\\&\leq \int_{0}^{\infty} \left(\frac{\alpha}{e\,y}\right)^\alpha q(x)\,dx\\&<\infty.
\end{align}
Similarly, 
\begin{IEEEeqnarray*}{rCl}
    \E{X\given Y=y}&\leq \frac{1}{\nu(y)}\left(\frac{\alpha+1}{e\,y}\right)^{\alpha+1}&<\infty\\
    \E{X^2\given Y=y}&\leq \frac{1}{\nu(y)}\left(\frac{\alpha+2}{e\,y}\right)^{\alpha+2}&<\infty.
\end{IEEEeqnarray*}
\end{remark}}
\section{Calculation of $\kappa_X$}
\label{kappacalc}

In~\cite{eswaran_gastpar_2019}, the lower bound is expressed as, for $0<D<\frac{N(X)}{\kappa_X}$
\begin{equation}
    L(D)\geq \frac{\sigma_W^2}{2} \left(\frac{1}{D}-\frac{\kappa_X}{N(X)}\right)-\frac{1}{2}\log \frac{\sigma_X^2}{N(X)}
\end{equation}
where
\begin{align}
    \kappa_X &= \lim_{s\rightarrow 0^+}\frac{d}{ds} N(X+\sqrt{s}G).
\end{align}
Under regularity conditions, $\kappa_X$ can be simplified to
\begin{align}
    \kappa_X &= \lim_{s\rightarrow 0^+}\frac{d}{ds} N(X+\sqrt{s}G)\\
    &= \lim_{s\rightarrow 0^+} 2N(X+\sqrt{s}G)\,\frac{d}{ds} h(X+\sqrt{s}G)\\
    &= 2N(X)\lim_{s\rightarrow 0^+} \frac{d}{ds} h(X+\sqrt{s}G)\\
    &= 2N(X) \,\lim_{s\rightarrow 0^+}\frac{1}{2}\, J(X+\sqrt{s}G)\\
    &= N(X)\, J(X) ,
\end{align}
which thus yields Equation~\eqref{Eqn-rateloss-lowerbound}.
\end{appendices}
\section*{Acknowledgment}
The authors are grateful to the Associate Editor and anonymous reviewers
for a meticulous reading of the draft and insightful feedback which helped improve the manuscript.

\bibliographystyle{IEEEtran}
\bibliography{main}


\end{document}